\documentclass[11pt,a4paper]{article}

\usepackage[a4paper,text={150mm,240mm},centering,headsep=10mm,footskip=15mm]{geometry}
\usepackage[T1]{fontenc}
\usepackage{epsfig, float,color}
\usepackage{amsmath,amssymb}
\usepackage{amsfonts}
\usepackage{url}
\usepackage{bbm}
\usepackage{mathrsfs} 

\newcommand{\D}{\dagger}

\newcommand{\R}{\mathbb{R}}
\newcommand{\C}{\mathbb{C}}
\newcommand{\N}{\mathbb{N}}

\newcommand{\id}{\mathbbm{1}}

\newcommand{\diag}{{\rm diag} }

\newtheorem{theorem}{Theorem}[section]
\newtheorem{lemma}[theorem]{Lemma}

\newenvironment{proof}[1][Proof:]{\begin{trivlist}
\item[\hskip \labelsep {\bfseries #1}]}{\end{trivlist}}

\newenvironment{remark}[1][Remark:]{\begin{trivlist}
\item[\hskip \labelsep {\bfseries #1}]}{\end{trivlist}}

\newcommand{\qed}{\nobreak \ifvmode \relax \else
      \ifdim\lastskip<1.5em \hskip-\lastskip
      \hskip1.5em plus0em minus0.5em \fi \nobreak
      \vrule height0.75em width0.5em depth0.25em\fi}

\title{A relativistically interacting exactly solvable multi-time model for two mass-less Dirac particles in $1+1$ dimensions}

\author{
Matthias Lienert\thanks{lienert@math.lmu.de, Mathematisches Institut, Ludwig-Maximilians-Universit\"at, 
	Theresienstr. 39, 80333 M\"unchen, Germany}
}

\date{\today}

\begin{document}

\maketitle

\begin{abstract}
	\noindent The question how to Lorentz transform an $N$-particle wave function naturally leads to the concept of a so-called multi-time wave function, i.e. a map from $({\rm space\text{-}time})^N$ to a spin space. This concept was originally proposed by Dirac as the basis of relativistic quantum mechanics. In such a view, interaction potentials are mathematically inconsistent. This fact motivates the search for new mechanisms for relativistic interactions. In this paper, we explore the idea that relativistic interaction can be described by boundary conditions on the set of coincidence points of two particles in space-time. This extends ideas from zero-range physics to a relativistic setting. We illustrate the idea at the simplest model which still possesses essential physical properties like Lorentz invariance and a positive definite density: two-time equations for mass-less Dirac particles in $1+1$ dimensions. In order to deal with a spatio-temporally non-trivial domain, a necessity in the multi-time 
picture, we develop a new method to prove existence and uniqueness of classical solutions: a generalized version of the method of characteristics. Both mathematical and physical considerations are combined to precisely formulate and answer the questions of probability conservation, Lorentz invariance, interaction and antisymmetry.\\
    
\end{abstract}

\section{Introduction}

The idea of a multi-time wave function was first suggested by Dirac in 1932 in his article ``Relativistic Quantum Mechanics'' \cite{dirac_32}. For him, working in the Schrödinger picture, it seemed to be the essential step from non-relativistic to relativistic quantum mechanics to replace the usual single-time wave function (here for $N$ particles and $d$ spatial dimensions)
\begin{equation}
 \varphi : \underbrace{\R^d \times \cdots \times \R^d}_{N~{\rm times}} \times \R ~\longrightarrow~\mathcal{S},~~~(\mathbf{x}_1,..., \mathbf{x}_N, t) ~\longmapsto~\varphi (\mathbf{x}_1,..., \mathbf{x}_N, t)
 \label{eq:singletimewavefn}
\end{equation}
with a relativistic wave function that involves a time coordinate for each particle:
\begin{equation}
 \psi : \underbrace{\R^{1+d} \times \cdots \times \R^{1+d}}_{N~{\rm times}} ~\longrightarrow~\mathcal{S},~~~(x_1,..., x_N) ~\longmapsto~\psi (x_1,..., x_N).
 \label{eq:multitimewavefn}
\end{equation}
Here, $\mathbf{x}_k$ denotes the spatial coordinates of the $k$-th particle, $x_k = (t_k,\mathbf{x}_k)$ its spatio-temporal coordinates and $\mathcal{S}$ a suitable spin space. Because of the presence of many time coordinates, these relativistic Schrödinger picture wave functions have been termed \textit{multi-time wave functions}. Their connection with single-time wave functions is straightforward:
\begin{equation}
 \varphi(\mathbf{x}_1,...,\mathbf{x}_N,t) ~=~ \psi(\mathbf{x}_1,t,...,\mathbf{x}_N,t).
 \label{eq:connection}
\end{equation}
 One can also think of a multi-time wave function as arising as follows: Consider a configuration of two particles $(\mathbf{x}_1, \mathbf{x}_2)$. In order to fully explain what is meant by such a configuration, one has to specify the time $t$ at which it is considered (in a certain Lorentz frame $F$): $(\mathbf{x}_1, \mathbf{x}_2, t)$. In order to transform this configuration to another frame $F'$, we have to consider the physically synonymous collection of $(1+d)$-vectors $q = ((t,\mathbf{x}_1),(t,\mathbf{x}_2))$. Only then one can apply the Lorentz transformation $\Lambda$ describing the transition between $F$ and $F'$: $q' = (\Lambda(t,\mathbf{x}_1),\Lambda(t,\mathbf{x}_2))$. However, this will in general yield $q' = (t_1',\mathbf{x}_1',t_2',\mathbf{x}_2')$ with $t_1' \neq t_2'$, i.e. not a simultaneous configuration. Thus, applying this consideration to the argument of a single-time wave function, we arrive at the \textit{necessity} to consider a multi-time wave function.\\
 Note that this exactly yields a multi-time wave function on the domain of \textit{space-like configurations.} To allow for time-like configurations in the argument of $\psi$ does not make sense, as they could actually correspond to points on the world line of a single particle. Relatedly, due to the expected statistical role of $\psi$, an adequate notion of normalization can only hold on space-like surfaces.\\ 
 As evolution equations, Dirac proposed a system of $N$ wave equations:
\begin{align}
 i \frac{\partial}{\partial {t_1}} \psi ~&=~ H_1 \psi\nonumber\\
 &~ \,\vdots\nonumber\\
  i \frac{\partial}{\partial {t_N}} \psi ~&=~ H_N \psi
  \label{eq:multitime}
\end{align}
where $H_k,~k=1,...,N$ are differential operators on an appropriate function space. The system \eqref{eq:multitime} is supposed to transform covariantly under the Lorentz group. As initial data one can e.g. choose to prescribe the wave function at a common time, say $t = 0 = t_1 = ... = t_N$ or, alternatively, a configuration on a general space-like hypersurface.\\
Mathematically, \eqref{eq:multitime} is particularly interesting because it is an overdetermined system of partial differential equations (PDEs). A common solution only exists if certain compatibility conditions are satisfied \cite{bloch,tomonaga,nogo_potentials}:
\begin{equation}
 \left[H_j - i \frac{\partial}{\partial t_j}, H_k - i \frac{\partial}{\partial t_k}\right] ~=~ 0~~\forall j,k.
\label{eq:kb}
\end{equation}
These conditions are in fact quite restrictive: For example, it has recently been shown in \cite{nogo_potentials} that they exclude potentials in the following sense: Let $H_j = H_j^{\rm free} + V_j$ where $H_j^{\rm free}$ is the free Dirac Hamiltonian (or Laplacian) acting on the coordinates and spin index of the $j$-th particle and $V_j$ is a matrix-valued function that may depend on all of the particle coordinates. Then the only $V_j$'s fulfilling eq. \eqref{eq:kb} are gauge equivalent to purely external potentials, i.e. choices of $V_j$ which only depend on $x_j$, not $x_k$  with $k \neq j$.\\
This result motivates the search for alternative mechanisms of interaction instead of potentials. In \cite{qftmultitime}, the multi-time approach is extended to quantum field theory, as originally proposed by Dirac, Fock and Podolsky \cite{dfp}, Tomonaga \cite{tomonaga} and Schwinger \cite{schwinger}. However, these approaches encounter difficulties with UV-divergencies which result, roughly speaking, from the back-reaction the field generated by a particle onto the particle itself \cite{feynman_nobel_lecture}. This renders the problem mathematically ill-defined and requires to introduce regularizing parameters so that problems with Lorentz invariance may arise.\\
Here we explain a different approach, inspired by the field of zero-range physics (see \cite{albeverio} for an overview). The main idea is that boundary conditions for the wave function may yield physically interesting effects and even interaction while the formal differential operator in the wave equation is the free one. This clearly avoids the use of potentials. However, the direct of the methods used in zero-range physics, developed for the single-time formalism, is not possible for multi-time equations.\\
To illustrate this claim, recall the standard functional-analytic treatment of single-time wave equations (see e.g. \cite{reed_simon_1,reed_simon_2} and \cite[chap. 14]{duerr})
\begin{equation}
 i \frac{\partial}{\partial t} \varphi ~=~ H \varphi
 \label{eq:singletime}
\end{equation}
where $H$ is a self-adjoint operator on a Hilbert space $\mathcal{H}$, most often $\mathcal{H} = L^2(\mathcal{Q}) \otimes \mathbb{C}^k$, where $\mathcal{Q} \subset \R^{Nd}$ is the physically accessible part of configuration space. Usually, $H$ is an unbounded operator with domain $\mathcal{D}(H) \subsetneq \mathcal{H}$. The specification of $\mathcal{D}(H)$ is important for physics, as it includes potential boundary conditions which influence time evolution and spectrum. $H$ is the generator of a strongly continuous unitary one-parameter group $U(t)=\exp(-i H t)$. If $\varphi_0 \in \mathcal{D}(H)$, then $U(t) \varphi_0 \in \mathcal{D}(H)$ and $U(t) \varphi_0$ satisfies eq. \eqref{eq:singletime}. The unitarity of $U$ ensures conservation of the norm of the wave function which is essential for the statistical meaning of the wave function.\\
It is crucial to note that within this approach, the allowed boundary conditions are time-less, as the Hilbert space $L^2(\mathcal{Q}) \otimes \mathbb{C}^k$ does not include time. Spatial boundary conditions prescribed in this way are automatically extended for all times $t$. An example for two particles is:
\begin{equation}
 \varphi(\mathbf{x}_1, \mathbf{x}_2,t) ~=~ 0~~~{\rm for} ~ \mathbf{x}_1 = \mathbf{x}_2,~\forall t.
 \label{eq:examplesingletime}
\end{equation}
For multi-time wave functions, the method can be generalized straightforwardly by using a theorem in \cite[thm. VIII.12]{reed_simon_1} on strongly continuous unitary $N$-parameter groups $U(t_1,...,t_N)$ on the same Hilbert space $\mathcal{H}$ as above. Such a group can be constructed if and only if the generators $H_j$ of the one-parameter subgroups $U(0,...,0,t_j,0,...,0)$ are self-ajoint, commute pairwise and have a common domain $\mathcal{D}(H_j) \equiv \mathcal{D}$, independent of $j$. This implies: if $\psi_0 \in \mathcal{D}$, then $U(t_1,...,t_N) \psi_0 \in \mathcal{D}$ and $U(t_1,...,t_N)\psi_0$ obeys the multi-time equations \eqref{eq:multitime}.\\
However, one crucial aspect changes: boundary conditions are still supposed to be expressed via the domain $\mathcal{D}$ which makes no reference to time. Consequently, the boundary conditions are automatically extended in \textit{all} coordinate times, e.g.:
\begin{equation}
 \psi(t_1,\mathbf{x}_1,t_2,\mathbf{x}_2) ~=~ 0~~~{\rm for} ~ \mathbf{x}_1 = \mathbf{x}_2,~\forall \, t_1,t_2.
 \label{eq:examplemultitime}
\end{equation}
Using the connection between single-time and multi-time equations (eq. \eqref{eq:connection}), one can see that condition \eqref{eq:examplemultitime} in fact differs from the corresponding one in the single-time formalism \eqref{eq:examplesingletime}, although one might have $\mathcal{D}(H) = \mathcal{D}$. Namely, eq. \eqref{eq:examplesingletime} translated into the multi-time formalism via \eqref{eq:connection} reads: 
\begin{equation}
 \psi(t_1,\mathbf{x}_1,t_2,\mathbf{x}_2) = 0~~~{\rm for} ~ \mathbf{x}_1 = \mathbf{x}_2,~ t_1 = t_2
 \label{eq:comparison}
\end{equation}
with condition ``${\rm for}~t_1 = t_2$'' instead of ``$\forall \, t_1, t_2$''. However, boundary conditions like \eqref{eq:examplemultitime} for spatio-temporal configurations which may be time-like do not have a clear meaning. It thus seems that the functional-analytic approach is not adequate for multi-time equations on domains with boundaries, since it automatically implements too many\footnote{It may well be that the only common domain $\mathcal{D}$ of self-adjointness of the $H_j$'s is the one corresponding to the free operators, i.e. one where no boundary condition such as \eqref{eq:comparison} is prescribed.} and physically unreasonable boundary conditions.
Therefore, a different method is required to implement the idea that boundary conditions could lead to relativistically invariant interaction for multi-time wave functions. In order for the boundary conditions to be Lorentz invariant, time should also be admitted in their formulation. We suggest to take a step back and view the multi-time equations \eqref{eq:multitime} as a general overdetermined system of PDEs on a subset of configuration space-time $\R^{N(1+d)}$, treating space and time on equal footing.\\
Of course, such a change in methods raises important questions, such as:
\begin{enumerate}
 \item How does one prove existence and uniqueness of solutions?
 \item How is probability conservation guaranteed and which notion thereof is adequate in a relativistic regime?
\end{enumerate}
(In the functional-analytic treatment, the two points are conveniently answered by the notion of self-adjointness.)\\
In this paper, we provide a model for which both questions can be answered definitely and precisely, bearing in mind also the physical aspects of interaction and Lorentz invariance. For this purpose, we consider a two-time system of mass-less Dirac equations in one spatial dimension ($d = 1$) on the domain of space-like configurations.\\
The choice of the model is explained as follows: The dimensionality both allows for an explicit solution in the mass-less case as well as leads to the situation that a certain natural Lorentz-invariant boundary in configuration space-time, the set of coincidence points, has the right dimensionality\footnote{In a functional-analytic setting, the dimensionality of the boundary to allow for zero-range interactions is known to depend sensitively on the order of the differential operator and the dimension of configuration space \cite{svendsen}.} to have impact on the time evolution. Moreover, the Dirac equation is Lorentz invariant, reflects the expected dispersion relation and possesses a conserved tensor current with a positive component that can play the role of a probability density. The choice of domain is explained by the considerations about the necessity of multi-time wave functions following eq. \eqref{eq:connection}. Interestingly, this immediately raises the question of boundary conditions since the 
domain of space-like configurations has a non-empty boundary: the light-like configurations. This provides a natural reason to study the idea of relativistic interaction by boundary conditions.

The paper is structured as follows: We begin with introducing the model, as defined by its multi-time equations, domain and initial conditions as well as boundary conditions at the space-time points of coincidence. Next, the general solution is found (lemma \ref{thm:generalsolution}) and existence and uniqueness of $C^k$-solutions are studied by a generalized method of characteristics (theorem \ref{thm:ibvp}). We continue with a proposal how probability conservation can be understood for multi-time wave functions (lemma \ref{thm:currentform}) and determine a general class of boundary conditions that guarantees it (theorem \ref{thm:currentconsbdyconds}). We proceed with proving the Lorentz invariance of the model, and particularly of the boundary conditions (lemma \ref{thm:libdyconds}). Moreover, a criterion for what constitutes interaction is suggested and applied to the model, showing that it is indeed interacting in this sense (theorem \ref{thm:interaction}). The time evolution and effect of the 
interaction are explicitly illustrated at the example of initially localized wave packets for each of the two particles. Finally, the implications of anti-symmetry for the boundary conditions in the case of indistinguishable particles are analyzed (lemma \ref{thm:antisymmetry}). We conclude with an outlook on possible generalizations of the model.

\section{The model}

Our model is based on a two-time wave function for two Dirac (spin-$\tfrac{1}{2}$) particles in $(1+1)$-dimensional space-time:
\begin{equation}
 \psi :~ \Omega \subset \R^2 \times \R^2 ~~\longrightarrow ~~ \C^2 \otimes \C^2, ~~~~~(t_1, z_1, t_2, z_2) ~~\longmapsto ~~ \psi(t_1, z_1, t_2, z_2).
 \label{eq:twotimewavefn}
\end{equation}
According to the arguments in the introduction, the physically natural choice of the domain $\Omega$ is the set $\mathscr{S}$ of space-like configurations, given by:
\begin{equation}
 \mathscr{S} ~:=~ \{ (t_1,z_1,t_2,z_2) \in \R^2 \times \R^2 : (t_1 - t_2)^2 - (z_1 - z_2)^2 < 0 \}.
 \label{eq:spacelikeconfigs}
\end{equation}
Our sign convention for the flat space metric is $g = \diag(1,-1)$.\\
Initial data should be prescribed on a surface $\mathcal{I}$ of the form $\mathcal{I} = (\Sigma_0 \times \Sigma_0) \cap \Omega$ where $\Sigma_0$ is a space-like hypersurface. We choose:
\begin{equation}
 \mathcal{I} ~:=~ \{ (t_1,z_1,t_2,z_2) \in \mathscr{S} : t_1 = t_2 = 0 \},
 \label{eq:initialdataset}
\end{equation}
i.e. a $\Sigma_0$ corresponding to $t = 0$.\\
In order to obtain a fully Lorentz invariant model, boundary conditions have to be prescribed on a Lorentz invariant subset of $\partial \Omega$. The first natural choice is the whole of $\partial \mathscr{S}$, i.e. the set
 \begin{equation}
  \mathscr{L} ~:=~ \{ (t_1,z_1,t_2,z_2) \in \mathscr{S} : (t_1 - t_2)^2 - (z_1 - z_2)^2 = 0 \}
  \label{eq:lightlighconfigs}
 \end{equation}
 of light-like configurations. However, this set has dimension three, as compared to dimension two of $\mathcal{I}$, so one expects that it leads to an overdetermined initial boundary value problem (IBVP). The second natural choice -- and the one we shall make -- is the set $\mathscr{C}$ of \textit{coincidence points} in space-time, given by:
\begin{equation}
 \mathscr{C} := \{ (t_1,z_1,t_2,z_2) \in \R^2 \times \R^2 : t_1 = t_2, z_1 = z_2 \}.
 \label{eq:coincidencepts}
\end{equation}
As two-time wave equations we use the free ($1+1$)-dimensional Dirac equations acting on the spin indices of the first and second particle, respectively:
\begin{align}
 i \frac{\partial}{\partial t_1} \psi(t_1, z_1, t_2, z_2) ~&= \, - i \, \sigma_3 \otimes \id_2 \frac{\partial}{\partial z_1} \, \psi(t_1, z_1, t_2, z_2), \nonumber\\
 i \frac{\partial}{\partial t_2} \psi(t_1, z_1, t_2, z_2) ~&= \, - i \, \id_2 \otimes \sigma_3 \frac{\partial}{\partial z_2} \,\psi(t_1, z_1, t_2, z_2).
 \label{eq:twotime}
\end{align}
In the case with mass, additional terms $m_1 \, \sigma_1 \otimes \id_2$ and $m_2 \, \id_2 \otimes \sigma_1$ appear in front of $\psi$ on the right hand side. Here,
\begin{equation}
 \sigma_1 = \left(\begin{array}{cc}
0 & 1 \\ 
1 & 0
\end{array} \right), ~~~
\sigma_2 = \left(\begin{array}{cc}
0 & -i \\ 
i & 0
\end{array} \right), ~~~
\sigma_3 = \left(\begin{array}{cc}
1 & 0 \\ 
0 & -1
\end{array} \right)
 \label{eq:pauli}
\end{equation}
are the Pauli matrices. Note that the compatibility conditions \eqref{eq:kb} are satisfied, as the matrices appearing in the first and second equation of \eqref{eq:twotime} are constant and commute.\\
To summarize, the model is given by:
\begin{equation}
 \left\{\begin{array}{l}
{\rm Eqs.}~\eqref{eq:twotime}~{\rm on}~\mathscr{S},\\
 \psi_i = g_i~{\rm on}~\mathcal{I},~i = 1,2,3,4,\\
 {\rm boundary~conditions~on}~\mathscr{C}. \end{array} \right.
 \label{eq:model}
\end{equation}
Here, $\psi_i,~i = 1,2,3,4$ denote the components of $\psi$ with respect to the (ordered) basis
\begin{equation}
 \mathcal{B} = ( e_1 \otimes e_1, e_1 \otimes e_2,e_2 \otimes e_1, e_2 \otimes e_2)
 \label{eq:psicomponents}
\end{equation}
where $e_i$ are the canonical basis vectors of $\C^2$. $g_i,~i = 1,2,3,4$ are arbitrary complex-valued $C^k$-functions on $\overline{\mathcal{I}}$. The form of admissible boundary conditions will be explored in the next section.\\
For future convenience note the following explicit representation for arbitrary complex-valued $2 \times 2$ matrices $A = (a_{ij}),~B = (b_{ij})$ with respect to $\mathcal{B}$:
\begin{equation}
 A \otimes \id_2 = \left(\begin{array}{cccc}
a_{11} & 0 & a_{12} & 0 \\ 
0 & a_{11} & 0 & a_{12} \\
a_{21} & 0 & a_{22} & 0 \\
0 & a_{21} & 0 & a_{22}
\end{array} \right), ~~~
\id_2 \otimes B = \left(\begin{array}{cccc}
b_{11} & b_{12} & 0 & 0 \\ 
b_{21} & b_{22} & 0 & 0 \\
0 & 0 & b_{11} & b_{12} \\
0 & 0 & b_{21} & b_{22} \\
\end{array} \right).
 \label{eq:tensorrep}
\end{equation}

\section{Existence and uniqueness}
In this section, it is shown which type of boundary conditions ensures existence and uniqueness of a $C^k$-solution (for any $k \in \N$) of the two-time equations. This is achieved using a generalized version of the method of characteristics.
\begin{lemma}
 On any open and connected domain $D \subset \R^2\times \R^2$, the general solution of the two-time system \eqref{eq:twotime}  is given by:
  \begin{equation}
    \left(\begin{array}{c} \psi_1\\ \psi_2\\ \psi_3\\ \psi_4 \end{array} \right)(t_1,z_1,t_2,z_2) ~=~ \left(\begin{array}{c} f_1(z_1-t_1,z_2-t_2)\\ f_2(z_1-t_1,z_2+t_2)\\ f_3(z_1+t_1,z_2-t_2)\\ f_4(z_1+t_1,z_2+t_2) \end{array} \right)
    \label{eq:generalsolution}
  \end{equation}
 where $f_j: \R^2 \rightarrow \C,~j=1,2,3,4$ are $C^1$-functions.
 \label{thm:generalsolution}
\end{lemma}
\begin{proof}
  Using eq. \eqref{eq:tensorrep}, we explicitly write out eq. \eqref{eq:twotime}:
  \begin{align}
  i \frac{\partial}{\partial t_1} \left(\begin{array}{c} \psi_1\\ \psi_2\\ \psi_3\\ \psi_4 \end{array} \right) ~&=~ - i \left(\begin{array}{cccc}
  1 & ~ & ~ & ~\\
  ~ & 1 & ~ & ~\\
  ~ & ~ & -1 & ~\\
  ~ & ~ & ~ & -1
  \end{array} \right) \frac{\partial}{\partial z_1} \left(\begin{array}{c} \psi_1\\ \psi_2\\ \psi_3\\ \psi_4 \end{array} \right), \nonumber\\
  i \frac{\partial}{\partial t_2} \left(\begin{array}{c} \psi_1\\ \psi_2\\ \psi_3\\ \psi_4 \end{array} \right) ~&=~ - i \left(\begin{array}{cccc}
  1 & ~ & ~ & ~\\
  ~ & -1 & ~ & ~\\
  ~ & ~ & 1 & ~\\
  ~ & ~ & ~ & -1
  \end{array} \right) \frac{\partial}{\partial z_2} \left(\begin{array}{c} \psi_1\\ \psi_2\\ \psi_3\\ \psi_4 \end{array} \right).
  \label{eq:twotimediag}
  \end{align}
  We see that the basis in spin space has been chosen such that all occurring matrices are diagonal. The structure of the equations becomes very simple. For example, for $\psi_1$ we have:
  \begin{equation}
    \left( \frac{\partial}{\partial t_1} + \frac{\partial}{\partial z_1} \right) \psi_1 = 0,~~~\left( \frac{\partial}{\partial t_2} + \frac{\partial}{\partial z_2} \right) \psi_1 = 0~~~\Rightarrow~ \psi_1(t_1,z_1,t_2,z_2) = f_1(z_1-t_1,z_2-t_2)
    \label{eq:psi1}
  \end{equation}
  where $f_1$ is $C^1$. The claim for the other components follows analagously. \qed
\end{proof}
It is instructive to understand this result in geometrical terms. Eq. \eqref{eq:generalsolution} implies that the components of the solution are constant along certain two-dimensional surfaces in $\R^2 \times \R^2$ (for some $c_1,c_2 \in \R$):
\begin{align}
 S_1(c_1,c_2) ~&:=~ \{ (t_1,z_1,t_2,z_2) \in \R^2 \times \R^2 : z_1 - t_1 = c_1,~z_2-t_2 = c_2\}\nonumber\\
 S_2(c_1,c_2) ~&:=~ \{ (t_1,z_1,t_2,z_2) \in \R^2 \times \R^2 : z_1 - t_1 = c_1,~z_2+t_2 = c_2\}\nonumber\\
 S_3(c_1,c_2) ~&:=~ \{ (t_1,z_1,t_2,z_2) \in \R^2 \times \R^2 : z_1 + t_1 = c_1,~z_2-t_2 = c_2\}\nonumber\\
 S_4(c_1,c_2) ~&:=~ \{ (t_1,z_1,t_2,z_2) \in \R^2 \times \R^2 : z_1 + t_1 = c_1,~z_2+t_2 = c_2\}
 \label{eq:char}
\end{align}
where the index refers to the component $\psi_i$ that is constant along $S_i$. This behavior closely resembles the method of characteristics (see e.g. \cite{kreiss,courant_hilbert_2}). We therefore call the surfaces $S_i$ \textit{multi-time characteristics}. They allow for a simple and powerful method to study the IBVP. Note that for the unbounded domain $\R^2 \times \R^2$ lemma \ref{thm:generalsolution} already yields existence and uniqueness of solutions for the initial value problem \eqref{eq:model}, with $f_i(x,y)$ from eq. \eqref{eq:generalsolution} given by $g_i(x,y)$. For more complex domains such as $\Omega = \mathscr{S}$, one has to know more about the topological structure (in particular connectedness).

\begin{lemma}
 The domain $\Omega$ is the is the disjoint union of the sets $\Omega_1$ and $\Omega_2$ where
 \begin{align} 
  \Omega_1 ~&:=~ \{ (t_1,z_1,t_2,z_2) \in \R^2 \times \R^2 : (t_1 - t_2)^2 - (z_1 - z_2)^2 < 0,~ z_1 < z_2 \},\nonumber\\
 \Omega_2 ~&:=~ \{ (t_1,z_1,t_2,z_2) \in \R^2 \times \R^2 : (t_1 - t_2)^2 - (z_1 - z_2)^2 < 0,~ z_1 > z_2 \}.
 \label{eq:omega12}
 \end{align}
 Furthermore, $\Omega_1$ and $\Omega_2$ cannot be connected by a curve lying entirely in $\Omega$.
 \label{thm:omega12}
\end{lemma}
\begin{proof}
 The first statement is obvious from the definition of $\Omega = \mathscr{S}$ (eq. \eqref{eq:spacelikeconfigs}). The second statement follows because $\Omega_1, \Omega_2$ are disjoint and open (as can be seen from eq. \eqref{eq:omega12}). \qed
\end{proof}
This splitting of $\Omega$ into path-wise disjoint parts implies that one should formulate the IBVP separately for $\Omega_1, \Omega_2$. In particular, this allows for more subtle boundary conditions as limits within either $\Omega_1$ or $\Omega_2$. To identify these limits would mean to reduce the number of possibilities to prescribing that $\psi$ has to be continuous across the boundary. It may, however, be physically reasonable to admit singularities (including jumps) of $\psi$ at the boundary. In fact, this situation is generic in the field of zero range physics \cite{albeverio} where these singularities for example appear for $\delta$-interactions.\\

\noindent Now we come to the main result of this section: the formulation of the initial boundary value problem and the corresponding proof of existence and uniqueness of solutions.

\begin{theorem} Let $k \in \N$. Given complex-valued $C^k$ functions $h_j^\pm$ as well as $g_i^{(j)}$ ($i = 1,2,3,4;~j = 1,2$) such that \eqref{eq:compatibvp} holds, there exists a unique solution $\psi$ which is $C^k$ on $\Omega_1$ and $\Omega_2$ for the initial boundary value problem defined by:
 \begin{enumerate}
  \item For $\Omega_1$:
  \begin{align}
  &\psi_i(0,z_1,0,z_2) ~=~ g_i^{(1)}(z_1,z_2),~~i = 1,2,3,4~{\rm for}~ z_1 < z_2~{\rm i.e.~ on}~\mathcal{I}_1 := \mathcal{I} \cap \Omega_1,\nonumber\\
  &\psi_3(t,z-0,t,z+0) ~=~ h^+_1(t,z)~{\rm for}~t \geq 0,~{\rm i.e.~on} ~\mathscr{C},\nonumber\\
  &\psi_2(t,z-0,t,z+0) ~=~ h^-_1(t,z)~{\rm for}~t < 0,~{\rm i.e.~on} ~\mathscr{C}.
  \label{eq:ibvp1}
 \end{align}
  \item For $\Omega_2$:
   \begin{align}
  &\psi_i(0,z_1,0,z_2) ~=~ g_i^{(2)}(z_1,z_2), ~i = 1,2,3,4 ~{\rm for}~ z_1 > z_2~{\rm i.e.~ on}~\mathcal{I}_2 := \mathcal{I} \cap \Omega_2,\nonumber\\
  &\psi_2(t,z+0,t,z-0) ~=~ h^+_2(t,z)~{\rm for}~t \geq 0,~{\rm i.e.~on} ~\mathscr{C},\nonumber\\
  &\psi_3(t,z+0,t,z-0) ~=~ h^-_2(t,z)~{\rm for}~t < 0,~{\rm i.e.~on} ~\mathscr{C}.
  \label{eq:ibvp2}
 \end{align}
 \end{enumerate}
Here,``$\pm 0$'' denotes the corresponding limits, e.g. $\psi(z-0,z+0) := \lim_{\varepsilon \rightarrow 0} \psi(z- \varepsilon, z+ \varepsilon)$.\\ Furthermore, it is required that the initial conditions satisfy the boundary conditions, i.e.:
\begin{align}
 g_3^{(1)}(z,z) ~&=~ h^+_1(0,z)~\forall z \in \R,\nonumber\\
 g_2^{(1)}(z,z) ~&=~ h^-_1(0,z)~\forall z \in \R,\nonumber\\
 g_2^{(2)}(z,z) ~&=~ h^+_2(0,z)~\forall z \in \R,\nonumber\\
 g_3^{(2)}(z,z) ~&=~ h^-_2(0,z)~\forall z \in \R
 \label{eq:compatibvp}
\end{align}
and also that these transitions between initial and boundary values be of regularity $C^k$.
 \label{thm:ibvp}
\end{theorem}

\begin{proof}
 We only show the statement for $\Omega_1$; the one for $\Omega_2$ follows analagously. The proof is structured as follows: First, we identify  the part of $\Omega_1$ where each component of $\psi$ is formally determined by initial data, i.e. their domain of dependence. Next, we check if there are also parts of $\Omega_1$ where the $\psi_i$ are not specified by initial data. We continue with demonstrating that the above-mentioned boundary conditions formally yield the missing values of the $\psi_i$. Subsequently, we make sure that the constructions actually work by explicitly demonstrating that there exist curves within the characteristic surfaces connecting each point in $\Omega$ with exactly one initial or boundary value. Finally, we write explicitly write down the solution in terms of initial data and show that it is indeed $C^k$.
 \begin{enumerate}
  \item \textit{Domain of dependence of the initial data:} Consider the initial conditions in \eqref{eq:ibvp1}. Using the general solution (eq. \eqref{eq:generalsolution}), we find:
    $\psi_i(0,z_1,0,z_2) = f_i(z_1,z_2) \stackrel{!}{=} g_i^{(1)}(z_1,z_2)$, $z_1 < z_2$. Formally, this equation determines $f_i = f_i(x,y)$ as a function on $\{ (x,y) \in \R^2 : x < y\}$. Geometrically, this means that the characteristics $S_i(x,y)$ intersect $\mathcal{I}$ in a single point $(0,x,0,y)_i$ for all $i$. Then $\psi_i$ is constant along $S_i(x,y)$. This consideration demonstrates uniqueness. However, existence is only guaranteed if one can connect the initial values with a continuous curve within $S_i$ that also remains in $\Omega_1$. This is shown under point 4.
  \item \textit{Complement of the domain of dependence of the initial data:}
    \begin{enumerate}
     \item $\psi_1(t_1,z_1,t_2,z_2) = f_1(z_1-t_1,z_2-t_2)$: We know from 1. that there exist points $(t_1,z_1,t_2,z_2) \in \Omega_1$ such that $z_1-t_1 < z_2-t_2$. However, is $z_1-t_2 > z_2-t_2$ also possible in $\Omega_1$? To answer this question, consider: $z_1-t_2 > z_2-t_2 \Leftrightarrow z_1 - z_2 \geq t_1-t_2$. In $\Omega_1$, $z_1-z_2 < 0$ which implies $t_1-t_2 < 0$ and therefore $|z_1-z_2| < |t_1 -t_2|$. This inequality states that the configuration $(t_1,z_1,t_2,z_2)$ has to be time-like, in contradiction with $\Omega_1 \subset \mathscr{S}$. So there are no points in $\Omega_1$ which require the function $f_1(x,y)$ to be defined for $x > y$. We proceed similarly for the other components.
     \item $\psi_4(t_1,z_1,t_2,z_2) = f_4(z_1+t_1,z_2+t_2)$: $z_1+t_1 > z_2+t_2 \Leftrightarrow t_1-t_2 > z_2-z_1$. Since $z_1 < z_2$ in $\Omega_1$, we obtain: $|t_1 - t_2| > |z_1-z_2|$, so also $f_4(x,y)$ is only required for $x < y$.
     \item $\psi_2(t_1,z_1,t_2,z_2) = f_2(z_1-t_1,z_2+t_2)$: $z_1-t_1 > z_2 + t_2 \Leftrightarrow -t_1 -t_2 > z_2 - z_1$. This time, the inequality can always be satisfied, e.g by choosing $z_1 < z_2$ arbitrarily and $t_1 = t_2 \equiv t$ with $t < (z_1-z_2)/2$. Thus, $f_2(x,y)$ is not yet determined fully by initial values. Note that this case appears only for $t_1 + t_2 < 0$.
     \item $\psi_3(t_1,z_1,t_2,z_2) = f_3(z_1+t_1,z_2-t_2)$: $z_1+t_1 > z_2-t_2 \Leftrightarrow t_1+t_2 > z_2-z_1$. Again, this can happen for all values of $z_1+t_2, z_2-t_2$, e.g. for $t_1 = t_2 \equiv t$ with $t > (z_2 - z_1)/2$. Note that this case requires $t_1 + t_2 > 0$.
    \end{enumerate}
  \item \textit{Domain of dependence of the boundary values:}
    \begin{enumerate}
     \item The condition $\psi_2(t,z-0,t,z+0) \stackrel{!}{=} h_1^-(t,z),~t < 0$ yields (leaving away the limit ``$\pm 0$'' for notational ease): $f_2(z-t,z+t) = h_1^-(t,z)$. Indeed, this determines the missing values $f_2(x,y),~x \geq y$ exactly once, as the map $\Phi: \{ (t,z) \in \R^2 : t < 0 \} \rightarrow \{ (x,y) \in \R^2 : x > y \},~(t,z) \mapsto (z-t,z+t)$ is bijective.
     \item Similarly, the condition $\psi_3(t,z-0,t,z+0)\stackrel{!}{=} h_1^+(t,z),~t \geq 0$ determines $f_3(x,y),~x\geq y$ exactly once as the map $\tilde{\Phi}: \{ (t,z) \in \R^2 : t \geq 0 \} \rightarrow \{ (x,y) \in \R^2 : x \geq y \},~(t,z) \mapsto (z+t,z-t)$ is bijective.
    \end{enumerate}
  \item \textit{Proof of existence:} We have to make sure that both initial values as well as boundary values can be transported along a mult-time characteristic while staying in $\Omega_1$. Then the above considerations show that the functions $f_i$ are determined uniquely.
  \begin{enumerate}
   \item For $\psi_1$: We have to show that there exists a continuous curve connecting $(t_1,z_1,t_2,z_2)$ with $(0,z_1-t_1,0,z_2-t_2)$ while staying within a multi-time characteristic $S_1$ defined by $z_1-t_1 = c_1,~z_2-t_2 = c_2$ and also in $\Omega_1$. In fact, such a path is given by:
   \begin{equation}
     \gamma_1: [0,1] \rightarrow S_1 \cap \Omega_1,~~~\gamma_1(\tau) :=  (\tau t_1, z_1 - t_1 + \tau t_1, \tau t_2, z_2 - t_2 + \tau t_2).        
   \end{equation}
   Obviously: $\gamma_1(0) = (0,z_1-t_1,0,z_2-t_2)$, $\gamma_1(1) = (t_1,z_1,t_2,z_2)$. Besides, the $z_1$-component of $\gamma_2(\tau)$ has to be smaller than the $z_2$-component: $z_1-t_1+ \tau t_1 < z_2 - t_2 + \tau t_2 \Leftrightarrow z_2 - z_1 > (t_2-t_1)(1-\tau)$. This inequality is satisfied because in $\Omega_1$, we have $z_1 < z_2$ and $|z_1-z_2| > |t_1-t_2|$.\\
   Furthermore, one has to ensure that $\gamma_1(\tau)$ always yields a space-like configuration. To see this, consider:
   \begin{align*}
    \tau^2(t_1 - t_2)^2 ~&<~ (z_1 - t_1 + \tau t_1 - z_2 + t_2 - \tau t_2)^2\nonumber\\
    \Leftrightarrow ~~~0 ~&<~ [(z_1-z_2) + (t_2-t_1)]^2 -2\tau (z_1-z_2)(t_2-t_1) - 2 \tau(t_2-t_1)^2.\nonumber
   \end{align*}
   Now we use $-(t_2-t_1)^2 > -(z_2-z_1)^2$ for the last summand which yields:
   \begin{align*}
    &[(z_1-z_2) + (t_2-t_1)]^2 -2\tau (z_1-z_2)(t_2-t_1) - 2 \tau(t_2-t_1)^2\\
     &~~~> ~[(z_1-z_2) + (t_2-t_1)]^2 -\tau [(z_1-z_2)^2 + 2 (z_1-z_2)(t_2-t_1) - (t_2-t_1)^2]\nonumber\\
      &~~~= ~[(z_1-z_2) + (t_2-t_1)^2](1-\tau).
   \end{align*}
   For $\tau \in (0,1)$ this is indeed greater than zero and for $\tau = 0,1$ the claim is evident, anyway.\\
   For the other components we only state the corresponding curves. The proof that they stay within $S_i \cap \Omega_1$ is analogous to the one above. In case of $\psi_2, \psi_3$ the curves start at boundary values, i.e. stay only within $S_i \cap \overline{\Omega}_1$.
   \item For $\psi_4$:
     \begin{align}
       \gamma_4: [0,1] \rightarrow S_4 \cap \Omega_1,~~~\gamma_4(\tau) :=  (\tau t_1, z_1 + t_1 - \tau t_1, \tau t_2, z_2 + t_2 - \tau t_2).
    \end{align}
   \item For $\psi_2$:
    \begin{align}
     &\gamma_2: [0,1] \rightarrow S_4 \cap \overline{\Omega}_1,\nonumber\\
     &\gamma_2(\tau) := \left\{  \begin{array}{l} 
				   \scriptstyle (\tau t_1, z_1 - t_1 + \tau t_1, \tau t_2, z_2 + t_2 - \tau t_2)~~~{\rm for}~z_1-t_1 < z_2 + t_2;\\
				   \scriptstyle \left( (-z_1 + z_2 + t_1 + t_2)/2 + \tau(z_1-z_2+t_1-t_2), (z_1+z_2-t_1+t_2)/2 + \tau(z_1-z_2+t_1-t_2),\right.\\
                                   \scriptstyle \left.(-z_1 + z_2 + t_1 + t_2)/2 + \tau(z_1-z_2-t_1+t_2), (z_1+z_2-t_1+t_2)/2 + \tau(-z_1+z_2+t_1-t_2) \right)\\
				   \scriptstyle ~{\rm for}~z_1-t_1 > z_2 + t_2.
                                     \end{array} \right.
    \end{align}
    The rather lengthy formula in the second case arises from a simple consideration. As before, $p = (t_1,z_1,t_2,z_2)$ is the point where we want to show the solution to be determined. Next, one determines the point $(t,z,t,z)$ of intersection of $S_2(z_1-t_1,z_2+t_2)$ with $\mathscr{C}$, obtaining $t = (-z_1+z_2+t_1+t_2)/2$ and $z = (z_1+z_2-t_1+t_2)/2$. Then: $\gamma_2(\tau) = (t+ \tau(t_1-t), z+\tau(z_1-z),t+\tau(t_2-t),z-\tau(z_2-z))$.
   \item For $\psi_3$:
    \begin{align}
     &\gamma_3: [0,1] \rightarrow S_3 \cap \overline{\Omega}_1,\nonumber\\
     &\gamma_3(\tau) := \left\{  \begin{array}{l} 
				   \scriptstyle (\tau t_1, z_1 + t_1 - \tau t_1, \tau t_2, z_2 - t_2 + \tau t_2)~~~{\rm for}~z_1+t_1 < z_2 - t_2;\\
				   \scriptstyle \left( (z_1-z_2+t_1+t_2)/2 + \tau(-z_1+z_2+t_1-t_2), (z_1+z_2+t_1-t_2)/2 + \tau(z_1-z_2-t_1+t_2),\right.\\
                                   \scriptstyle \left.(z_1 - z_2 + t_1 + t_2)/2 + \tau(-z_1+z_2-t_1+t_2), (z_1+z_2+t_1-t_2)/2 + \tau(-z_1+z_2-t_1+t_2) \right)\\
				   \scriptstyle ~{\rm for}~z_1+t_1 > z_2 - t_2.
                                  \end{array} \right.
    \end{align}
    The expression in the second case results from an analogous consideration as for $\psi_2$, the only change being the use of the point $(t,z,t,z)$ of intersection of $S_3(z_1+t_1,z_2-t_2)$ with $\mathscr{C}$.
  \end{enumerate}
  \item \textit{Explicit solution and $C^k$ property:} Collecting the results from the previous points, we obtain on $\Omega_1$:
     \begin{align}
      \psi_1(t_1,z_1,t_2,z_2)~&=~g_1^{(1)}(z_1-t_1,z_2-t_2),\nonumber\\
      \psi_2(t_1,z_1,t_2,z_2)~&=~\left\{ \begin{array}{l}
                                          \scriptstyle g_2^{(1)}(z_1-t_1,z_2+t_2)~~~{\rm for}~z_1-t_1 < z_2 + t_2\\
					  \scriptstyle h_1^-((-z_1 + z_2 + t_1 + t_2)/2,(z_1+z_2-t_1+t_2)/2)~~~{\rm for}~z_1-t_1 \geq z_2 + t_2
                                        \end{array}, \right.\nonumber\\
      \psi_3(t_1,z_1,t_2,z_2)~&=~\left\{ \begin{array}{l}
                                          \scriptstyle g_3^{(1)}(z_1+t_1,z_2-t_2)~~~{\rm for}~z_1+t_1 < z_2 - t_2\\
					  \scriptstyle h_1^+((z_1-z_2+t_1+t_2)/2,(z_1+z_2+t_1-t_2)/2)~~~{\rm for}~z_1+t_1 \geq z_2 - t_2
                                        \end{array}, \right.\nonumber\\
      \psi_4(t_1,z_1,t_2,z_2)~&=~g_4^{(1)}(z_1+t_1,z_2+t_2).
      \label{eq:explicitsolution}
     \end{align}
    From this formula, we immediately see that $\psi_1,\psi_4$ are $\C^k$ on $\Omega_1$ as $g_1^{(1)},g_2^{(1)}$ are $C^k$. For the other two components a similar argument holds true if additionally the transition between the two cases is $C^k$.
    \begin{enumerate}
     \item For $\psi_2$: The critial points are at $u := z_1-t_1 = z_2 + t_2$. We obtain as a condition that
      \begin{equation}
       g_2^{(1)}(u,u) ~\stackrel{!}{=}~ h_1^-(0,u)~\forall u \in \R
      \end{equation}
      and that this transition be $C^k$. We recognize this as one of the conditions in \eqref{eq:compatibvp} in the statement of the theorem.
     \item For $\psi_3$: The critial points are at $v := z_1+t_1 = z_2 - t_2$. We obtain as a condition that
      \begin{equation}
       g_3^{(1)}(v,v) ~\stackrel{!}{=}~ h_1^+(0,v)~\forall v \in \R
      \end{equation}
      and that this transition be $C^k$. Again, this is one of the conditions in \eqref{eq:compatibvp}. \qed
    \end{enumerate}
 \end{enumerate}
\end{proof}

\begin{remark}
 Note that for the definition of the functions $h_j^\pm$ one can make use of those components of $\psi$ that are already determined by initial values at the boundary point $(t,z,t,z)$ in question.
\end{remark}

\section{Probability conservation}

In this section we first introduce relativistic notation in order to formulate an adequate relativistic notion of probability conservation. A geometrical picture involving a $2d$-form is developed which enables us to prove the main result of this section: we identify a general class of conditions on the tensor current of the theory under which probability conservation is guaranteed. It is shown that these conditions are equivalent to certain linear relations between the components of the wave function that are covered by theorem \ref{thm:ibvp}.

\subsection{Relativistic notation}
We denote coordinates of particles in $\R^{1+d}$ by $x_i := (t_i,\mathbf{x}_i),~i = 1,2,...,N$. Their components are called $x_i^\mu,~\mu = 0,..,d$. Here, $N = 2$ and $d = 1$. Then: $x_i = (t_i,z_i)$. Partial derivatives with respect to $x_i^\mu$ are abbreviated by $\partial_{i,\mu}$. The Dirac gamma matrices are denoted by $\gamma^\mu$. They satisfy the Clifford algebra relations
\begin{equation}
 \gamma^\mu \gamma^\nu + \gamma^\nu \gamma^\mu ~=~ 2 g^{\mu \nu} \, \id,~~~\mu,\nu = 0,...,d.
 \label{eq:clifford}
\end{equation}
For $d = 1$, these are $2\times 2$-matrices. We choose the following representation:
\begin{equation}
 \gamma^0 = \sigma_1,~~~\gamma^1 = \sigma_1 \sigma_3.
 \label{eq:gamma1d}
\end{equation}
$\gamma_i^\mu$ stands for the $\mu$-th gamma matrix acting on the spin index of the $i$-th particle, i.e.:
\begin{equation}
 \gamma_1^\mu = \gamma^\mu \otimes \id,~~~\gamma_2^\nu = \id \otimes \gamma^\nu.
 \label{eq:gamma}
\end{equation}
Using this notation, we can rewrite the two-time system \eqref{eq:twotime} as:
\begin{equation}
 i \gamma_k^\mu \partial_{k,\mu} \, \psi(x_1,x_2) = 0,~~~k = 1,2
 \label{eq:reltwotime}
\end{equation}
where summation over upper and lower Greek indices is understood.\\
Let $\overline{\psi} := \psi^\D \gamma_1^0 \gamma_2^0$ denote the Dirac adjoint for two particles. Here, $\psi^\D$ stands for the conjugate transposed of $\psi$. Then eq. \eqref{eq:reltwotime} and the corresponding equation for $\overline{\psi}$ imply continuity equations for the tensor current $j$, defined by:
\begin{align}
 &j^{\mu \nu}(x_1,x_2) ~=~ \overline{\psi}(x_1,x_2) \gamma_1^\mu \gamma_2^\nu \psi(x_1,x_2) \label{eq:j},\\
 {\rm i.e.}~~~&\partial_{1,\mu} j^{\mu \nu}(x_1,x_2)  ~=~ \partial_{2,\nu} j^{\mu \nu}(x_1,x_2) ~=~ 0. \label{eq:continuity}
\end{align}
Note that $j^{00} = \psi^\D \psi$ yields the usual $|\psi|^2$ probability density.

\subsection{A relativistic notion of probability conservation}
In order to find an adequate relativistic notion of probability conservation, consider the usual non-relativistic notion:
\begin{equation}
 \int d^d x_1 \int d^d x_2  ~ |\psi|^2(t,\mathbf{x}_1, t,\mathbf{x}_2) ~=~ 1,~~~{\rm independent~of~t}.
 \label{eq:probcons1}
\end{equation}
We can rewrite this equation using $j^{00} = |\psi|^2$, making the geometric structure explicit:
\begin{equation}
 \int_{\Sigma_t} d\sigma(x_1) \int_{\Sigma_t} d\sigma(x_2)  ~ j^{00}(x_1,x_2) ~=~ 1,~~~{\rm independent~of~t}
 \label{eq:probcons2}
\end{equation}
where $\Sigma_t := \{ (\tau,\mathbf{x}) \in \R^{1+d} : \tau = t\}$.\\
It is now easily recognized that a special family of hypersurfaces, the equal time surfaces $\Sigma_t$ in a distinguished Lorentz frame, are used in the non-relativistic formulation. This flaw can be overcome by demanding the corresponding condition for all space-like hypersurfaces\footnote{Throughout the paper we assume that space-like hypersurfaces are smooth and possess a normal covector field at every point.} $\Sigma$. Let $n$ denote the normal covector field at $\Sigma$. We propose the following condition\footnote{A related idea is used in \cite[p. 163]{dirk_phd} to define $N$-particle Hilbert spaces associated with a space-like hypersurface $\Sigma$.}:
\begin{equation}
 \int_{\Sigma} d\sigma_1(x_1) \int_{\Sigma} d\sigma_2(x_2)  ~ n_\mu(x_1) n_\nu(x_2) \, j^{\mu \nu}(x_1,x_2) ~=~ 1,~~~{\rm independent~of~}\Sigma.
 \label{eq:probcons}
\end{equation}
This is justified as follows: Firstly, the condition is completely geometric and does not attribute significance to a special class of space-like hypersurfaces. Secondly, for $\Sigma_t$ one has $n \equiv (1,0,...,0)$, so eq. \eqref{eq:probcons} correctly reduces to eq. \eqref{eq:probcons2}. Thirdly, the meaning of eq. \eqref{eq:probcons} as expressing probability conservation can be established rigorously by a relativistic Bohmian analysis (see \cite{hbd,hbd_subsystems}).\\
In the case of a domain $\Omega \subset \R^{N(1+d)}$ with boundary, such as $\mathscr{S}$, one should restrict the range of integration to values in the domain and use the condition
\begin{equation}
 \int_{(\Sigma \times \Sigma)\cap \Omega} d\sigma_1(x_1) \wedge d\sigma_2(x_2)  ~ n_\mu(x_1) n_\nu(x_2) \, j^{\mu \nu}(x_1,x_2) ~=~ 1,~~~{\rm independent~of~}\Sigma.
 \label{eq:probconsbound}
\end{equation}
The idea is to employ Stokes' theorem to determine the conditions on $j$ such that probability conservation in the sense of eq. \eqref{eq:probconsbound} is guaranteed. To this end, it is useful to recognize
\begin{equation}
 \omega_j ~:=~ d\sigma_1 \wedge d\sigma_2  \,n_\mu n_\nu \, j^{\mu \nu}
 \label{eq:defcurrentform}
\end{equation}
as an $N  d$-form\footnote{See \cite[chap. 16.1]{duerr} for a similar idea for the non-relativistic case.}. In order to express this \textit{current form} by the coordinate differentials $d x_i^\mu$, we make use of the following results \cite[p. 435]{koenigsberger2}:
\begin{align}
 d \sigma_i(x_i) ~&=~ \sum_{\mu = 0}^d (-1)^\mu n_\mu(x_i) \, d x_i^0 \wedge \cdots \widehat{d x_i^\mu} \cdots \wedge d x_i^d,\label{eq:differentialforms1}\\
 n_\mu \, d \sigma_i ~&=~ (-1)^\mu \, d x_i^0 \wedge \cdots \widehat{d x_i^\mu} \cdots \wedge d x_i^d
 \label{eq:differentialforms2}
\end{align}
where $\widehat{d x_i^\mu}$ means that the corresponding factor should be omitted from the wedge product.
Using eq. \eqref{eq:differentialforms2} in the expression for $\omega_j$, we obtain:
\begin{lemma}
 \begin{enumerate}
  \item The current form can be rewritten as
    \begin{align}
     \omega_j ~=~ \sum_{\mu,\nu = 0}^{d} (-1)^\mu (-1)^\nu j^{\mu \nu} \, d x_1^0 \wedge \cdots \widehat{d x_1^\mu} \cdots \wedge d x_1^d ~ \wedge ~ dx_2^0 \wedge \cdots \widehat{d x_2^\nu} \cdots \wedge d x_2^d.
      \label{eq:currentform}
      \end{align}
  \item Probability conservation on domains $\Omega \subset \R^{2(1+d)}$ with boundary can be expressed by the following condition on the current form:
\begin{align}
 &\int_{(\Sigma \times \Sigma) \cap \Omega} \omega_j~=~ 1,~~~{\rm independent~of~}\Sigma.
 \label{eq:probconscrit}
\end{align}
 \end{enumerate}
 \label{thm:currentform}
\end{lemma}
 The continuity equations for $j$ yield:
\begin{lemma}
 The exterior derivative of $\omega_j$ vanishes, i.e. $d \omega_j = 0$.
 \label{thm:omegajexact}
\end{lemma}
\begin{proof}
 \begin{align}
 d \omega_j ~&=~ \sum_{\mu,\nu = 0}^{d} (-1)^\mu (-1)^\nu \partial_{1,\mu} j^{\mu \nu} \, (-1)^\mu \, d x_1^0 \wedge \cdots \wedge d x_1^d ~ \wedge ~ dx_2^0 \wedge \cdots \widehat{d x_2^\nu} \wedge \cdots \wedge d x_2^d \nonumber\\
 &~ \, + \sum_{\mu,\nu = 0}^{d} (-1)^\mu (-1)^\nu \partial_{2,\nu} j^{\mu \nu} \, (-1)^{3-\nu} \, d x_1^0 \wedge \cdots \widehat{d x_1^\mu} \cdots \wedge d x_1^d ~ \wedge ~ dx_2^0 \wedge \cdots \wedge d x_2^d \nonumber\\
 &\stackrel{\rm eq. \eqref{eq:continuity}}{=}~ 0. \qed
 \label{eq:domega}
\end{align}
\end{proof}
This result will allow us to relate the hypersurface integrals in \eqref{eq:probconscrit} using Stokes' theorem.

\subsection{Boundary conditions derived from probability conservation}
With criterion \eqref{eq:probconscrit} and the tools developed in the last section, we are almost ready to identify probability-conserving boundary conditions. Before stating the main result, we formulate a lemma that allows us to control the spreading of the wave function.
\begin{lemma}
 Consider the IBVP defined by \eqref{eq:ibvp1}, \eqref{eq:ibvp2} and let $\Sigma$ denote a space-like hypersurface. Then, if the initial data are compactly supported on $\mathcal{I}$, they are compactly supported on all sets of the form $(\Sigma \times \Sigma) \cap \Omega$.
 \label{thm:compactsupp}
\end{lemma}
\begin{proof}
 This can be seen immediately from the explicit solution \eqref{eq:explicitsolution}. (Influences propagate with finite speed along the multi-time characteristics.) \qed
\end{proof}
\begin{theorem}
  Let $\varepsilon_{\mu \nu}$ denote the Levi-Civita symbol. Assume furthermore that the initial data are of regularity $C^k,~k \in \N,$ and compactly supported on $\mathcal{I}$. Then the following conditions for the tensor current guarantee probability conservation in the sense of criterion \eqref{eq:probconscrit}:
 \begin{align}
  \varepsilon_{\mu \nu} j^{\mu \nu}(t,z-0,t,z+0) ~&\stackrel{!}{=} ~0,~~t,z \in \R, \nonumber\\
  \varepsilon_{\mu \nu} j^{\mu \nu}(t,z+0,t,z-0) ~&\stackrel{!}{=}~ 0,~~t,z \in \R.
  \label{eq:currentconsconds}
 \end{align}
 Expressed in terms of the components of $\psi$, these conditions are equivalent to:
 \begin{align}
 \psi_2(t,z-0,t,z+0) ~&\stackrel{!}{=}~ e^{-i \theta_1(t,z)} \psi_3(t,z-0,t,z+0),~~t,z \in \R,\nonumber\\
 \psi_2(t,z+0,t,z-0) ~&\stackrel{!}{=}~ e^{-i \theta_2(t,z)} \psi_3(t,z+0,t,z-0),~~t,z \in \R
  \label{eq:currentconsbdyconds}
 \end{align}
 for arbitrary functions $\theta_1, \theta_2: \R^2 \rightarrow [-\pi,\pi)$. (In order for $\psi$ to be $C^k$, they have to be $C^k$-functions, too.)
 \label{thm:currentconsbdyconds}
\end{theorem}
\begin{remark}
\begin{enumerate}
 \item Conditions \eqref{eq:currentconsconds} have the physical meaning that the probability flux from $\Omega_1$ into $\mathscr{C}$ and from $\Omega_2$ into $\mathscr{C}$ has to vanish, separately. They are therefore a subclass of all conditions on $j$ that lead to probability conservation. Boundary conditions with a similar meaning are widely used to express confinement of particles in certain spatial regions (see e.g. \cite[chaps. 12,14]{duerr} for a physically motivated discussion). The crucial difference here is that the boundary set is determined by internal relations between the particles, not by external geometry.
 \item Note that the boundary conditions \eqref{eq:currentconsbdyconds} are of the form \eqref{eq:ibvp1}, \eqref{eq:ibvp2} with property \eqref{eq:compatibvp}. Thus, theorem \ref{thm:ibvp} ensures existence and uniqueness of a $C^k$-solution on $\Omega_1$ and $\Omega_2$ of the corresponding IBVP.
\end{enumerate}
\end{remark}
\begin{proof}
 The idea is to use Stokes' theorem for a closed surface $S$ of the form $S = [(\Sigma_1 \times \Sigma_1) \cap \Omega] \cup [(\Sigma_2 \times \Sigma_2) \cap \Omega] \cup M$ where $\Sigma_1, \Sigma_2$ are space-like hypersurfaces. Then, because of $d \omega_j = 0$, one obtains equality of the normalization integrals \eqref{eq:probconsbound} if the contribution of $M$ vanishes. Parts of the contribution of $M$ vanish because $\psi$ is compactly supported on sets of the form $(\Sigma \times \Sigma) \cap \Omega$ according to lemma \ref{thm:compactsupp}. Demanding that the remaining parts also vanish leads to conditions on the tensor current.\\
 We split the proof into two parts: the first one to establish the conditions on the current such that the normalization integral of the wave function is equal for all hypersurfaces $\Sigma$ and the second one to derive the equivalent conditions for the components of $\psi$.
 \begin{enumerate}
  \item We first show that $S$ can be understood as a closed surface in an appropriate sense. To this end, we define finite versions of $\Sigma_1, \Sigma_2$. Pick $p_1 \in \Sigma_1$, $p_2 \in \Sigma_2$. Then let
  \begin{equation}
   \Sigma_i^R ~:=~ \{ p \in \Sigma_i : -(p^0-p_i^0)^2 + (\textbf{p}-\textbf{p}_i)^2 < R^2\},~~~i = 1,2.
   \label{eq:cutoffhypersurf}
  \end{equation}
 For $R$ large enough and $q \in [(\Sigma_i \backslash \Sigma_i^R) \times (\Sigma_i \backslash \Sigma_i^R) ]\cap \Omega $ we have $\psi(q) = 0$ as $\psi$ is compactly supported on sets of the form $(\Sigma \times \Sigma) \cap \Omega$. Consequently, one obtains
 \begin{equation}
  \int_{(\Sigma_i \times \Sigma_i) \cap \Omega} \omega_j ~=~ \int_{(\Sigma_i^R \times \Sigma_i^R) \cap \Omega} \omega_j,~~i = 1,2.
  \label{eq:sigmarnormalization}
 \end{equation}
 It is therefore permitted to replace $\Sigma_i$ with $\Sigma_i^R$ for the purpose of the argument.\\
 Now we construct a closed surface $S_R$ as follows: Let $V_{\Sigma_1, \Sigma_2} \subset \R^{1+d}$ be the volume between $\Sigma_1, \Sigma_2$, i.e. if $t_\Sigma(\mathbf{x})$ denotes the time coordinate of the unique point $p \in \Sigma$ with spatial coordinates $\mathbf{x}$, then
 \begin{equation}
  V_{\Sigma_1, \Sigma_2}~:=~ \{ (\tau, \mathbf{y}) \in \R^{1+d} : t_{\Sigma_1}(\mathbf{y}) < \tau < t_{\Sigma_2}(\mathbf{y}) \, \vee \,  t_{\Sigma_2}(\mathbf{y}) < \tau < t_{\Sigma_1}(\mathbf{y}) \}.
  \label{eq:enclosedvolume}
 \end{equation}
 Next, consider a continuous deformation of $\Sigma_1^R$ into $\Sigma_2^R$ (see fig. \ref{fig:continuousdeformation}), i.e. a smooth map
 \begin{equation}
  \Phi: [0,1] \rightarrow \{ \Sigma \subset \overline{V}_{\Sigma_1, \Sigma_2} : \Sigma~{\rm space\text{-}like~surface}\},~~~{\rm with}~\Phi(0) = \Sigma_1^R,~\Phi(1) = \Sigma_2^R.
  \label{eq:continuousdef}
 \end{equation}
\begin{figure}[bt]
   \centering
   \includegraphics[width=0.6\textwidth]{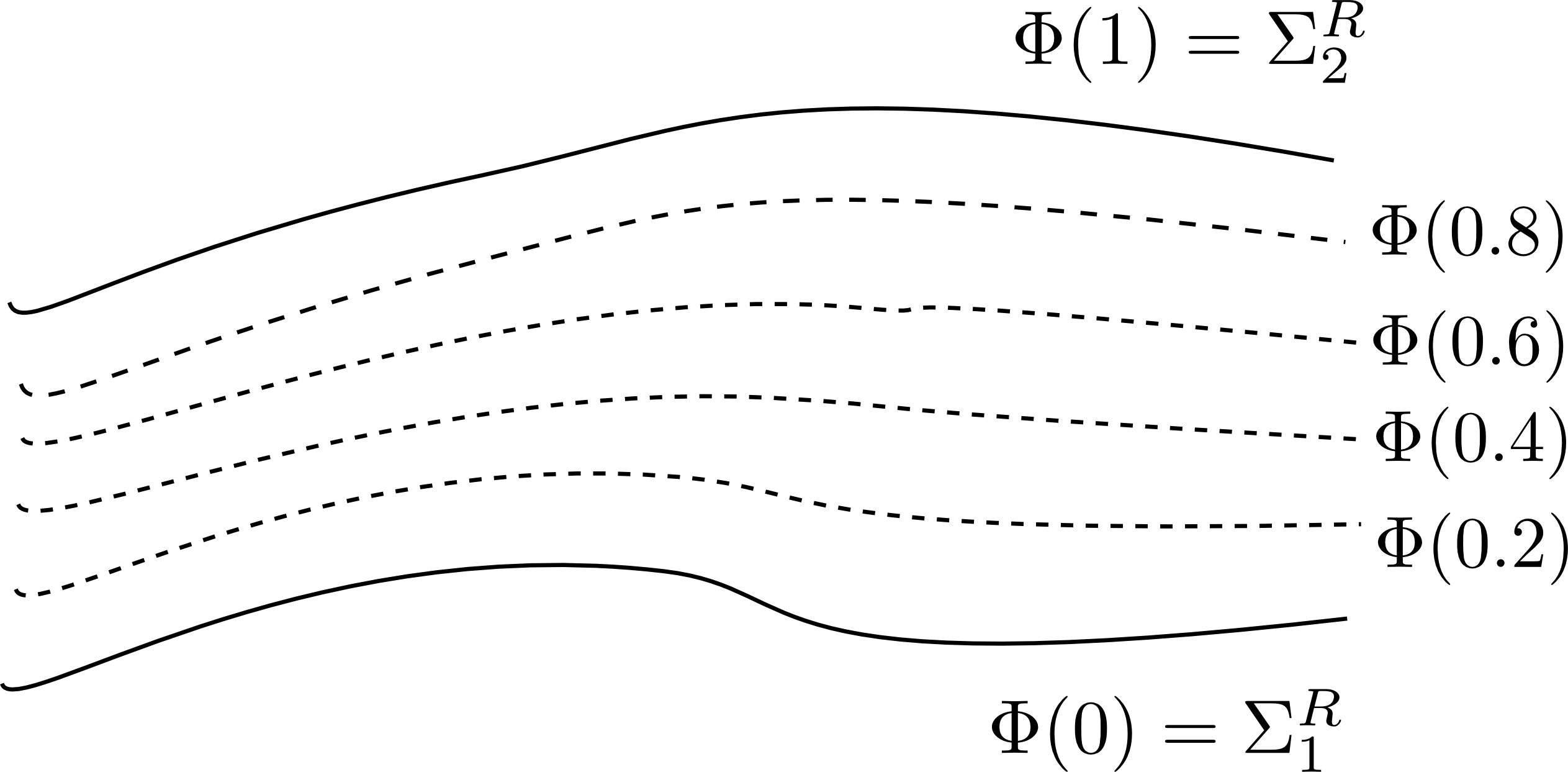}
   \caption{A continuous deformation of $\Sigma_1^R$ into $\Sigma_2^R$.}
 \label{fig:continuousdeformation}
 \end{figure}
 Now let
 \begin{equation}
  S_R~:=~ \partial \left( \bigcup_{s \in [0,1]} [\Phi(s) \times \Phi(s)] \cap \Omega \right).
  \label{eq:sr}
 \end{equation}
 By construction, $S_R$ is a closed surface. It has the form
 \begin{equation}
  S_R ~=~ [(\Sigma_1^R \times \Sigma_1^R) \cap \Omega] \cup [(\Sigma_2^R \times \Sigma_2^R) \cap \Omega] \cup M_1 \cup M_2
  \label{eq:srform}
 \end{equation}
 where $\psi \equiv 0$ on $M_2$. From eq. \eqref{eq:sr} it can be seen that $M_1$ consists of those points $p \in \Phi(s) \times \Phi(s)$ which do not lie in $\Omega$. As $\Phi(s)$ is a space-like surface, all points $p = (x,y),~x,y \in \Phi(s)$ with $x\neq y$ are contained in $\Omega$. The remaining ones therefore belong to the set $\mathscr{C}$ of coincidence points and it follows that $M_1 \subset \mathscr{C}$.\\
 At this point, a subtlety appears: Recall that the values of $\psi$ are not defined on $\mathscr{C}$ (as $\mathscr{C} \nsubseteq \Omega$). Rather, one has to consider the corresponding limits in $\Omega_1$ and $\Omega_2$. Instead of $S_R$, one should consider the union of $S_R^{(1)}$ with $S_R^{(2)}$ where $S_R^{(i)},~i=1,2$ are defined by eq. \eqref{eq:sr} using $\Omega_i,~i= 1,2$ instead of $\Omega$. They have the form
  \begin{equation}
  S_R^{(i)} ~=~ [(\Sigma_1^R \times \Sigma_1^R) \cap \Omega_i] \cup [(\Sigma_2^R \times \Sigma_2^R) \cap \Omega_i] \cup M_1 \cup M_2^{(i)}
  \label{eq:srform2}
 \end{equation}
 where $M_1$ is the same as above and $\psi \equiv 0$ on $M_2^{(i)},~i=1,2$.\\
 Let $V_R$ denote the volume enclosed by $S_R$. Using Stokes' theorem for $S_R$, we obtain:
 \begin{align}
  \int_{S_R} \omega_j ~&=~ \int_{V_R} d \omega_j ~\stackrel{\rm eq.\eqref{eq:domega}}{=}~0\nonumber\\
 \Rightarrow~~~\int_{(\Sigma_1^R \times \Sigma_1^R) \cap \Omega} \omega_j ~&=~ \int_{(\Sigma_2^R \times \Sigma_2^R) \cap \Omega} \omega_j ~- \int_{M_1} \omega_j^{(1)} ~+ \int_{M_1} \omega_j^{(2)}
  \label{eq:stokes1}
 \end{align}
 where $\omega_j^{(i)}$ is shorthand for taking the limit $\varepsilon \rightarrow 0$ for $\psi(t,z+(-1)^{i}\varepsilon,t,z-(-1)^{i}\varepsilon)$ in the expression for $\omega_j$. Orientation conventions have to be considered to obtain the correct signs in front of the integrals.\\
 Thus, we obtain independence of the normalization integrals from $\Sigma$ if
 \begin{equation}
  \int_{M_1} \omega_j^{(1)} ~=~ \int_{M_1} \omega_j^{(2)}.
  \label{eq:normindptcond1}
 \end{equation}
 We specialize to the case\footnote{The general case would lead to compensating currents from $\Omega_1$ to $\Omega_2$ and the other way around. This would mean that the particles could pass each other -- which we regard as physically questionable in $d = 1$.}
 \begin{equation}
  \int_{M_1} \omega_j^{(1)} ~=~ \int_{M_1} \omega_j^{(2)} ~=~0.
  \label{eq:normindptcond2}
 \end{equation}
 This condition will be satisfied if the current form obtained from the corresponding limit vanishes on $M_1$, or more generally on $\mathscr{C}$. The latter is reasonable to demand to make the construction work for any $\Sigma_1,\Sigma_2$.\\
 So far, the construction works for any dimension $d$. We now specialize to $d = 1$. In order to obtain an appropriate condition on $j$, we express $\omega_j$ using relative coordinates
 \begin{align}
 &z = z_1 - z_2,~~~Z = z_1 + z_2,~~~\tau = t_1 - t_2,~~~T = t_1 + t_2\nonumber\\
 \Leftrightarrow~~~&z_1 = \tfrac{1}{2}(Z+z),~~~z_2 = \tfrac{1}{2}(Z-z),~~~t_1 = \tfrac{1}{2}(T+\tau),~~~t_2 = \tfrac{1}{2}(T-\tau).
 \label{eq:relcoords}
\end{align}
This yields:
\begin{align}
 \omega_j ~&=~ \tfrac{1}{2} j^{00} dz \wedge dZ - \tfrac{1}{4}(j^{10}+j^{01}) d \tau \wedge dZ + \tfrac{1}{4}(j^{10}-j^{01}) d\tau \wedge dz \nonumber\\
 &~~~ -\tfrac{1}{4}(j^{10}-j^{01}) d T \wedge dZ - \tfrac{1}{4}(j^{10}+j^{01}) dz \wedge dT + \tfrac{1}{2} j^{11} d \tau \wedge dT.
\label{eq:omegarelcoords}
\end{align}
Now, on $\mathscr{C}$ we have $\tau = 0, z = 0$. Thus, we find:
\begin{equation}
 \omega_j(t,z+(-1)^{i}0,t,z-(-1)^{i}0) ~=~ \tfrac{1}{4} (j^{01}-j^{10})(t,z+(-1)^{i}0,t,z-(-1)^{i}0) \, dT \wedge dZ
 \label{eq:condm12}
\end{equation}
which leads to the following condition for the tensor current:
\begin{equation}
 (j^{01}-j^{10})(t,z+(-1)^{i}0,t,z-(-1)^{i}0) ~\stackrel{!}{=} ~0,~~i = 1,2.
 \label{eq:condj}
\end{equation}
 Recalling $\varepsilon_{\mu\nu} j^{\mu \nu} = j^{01}-j^{10}$, one easily verifies that these are exactly the conditions stated in eq. \eqref{eq:currentconsconds}.
 \item Next, we show that the conditions for the current (which are bilinear in $\psi$) are actually equivalent to the linear relations between the components of $\psi$ stated in eq. \eqref{eq:currentconsbdyconds}. For this purpose, consider
 \begin{align}
   j^{01} - j^{10} ~&=~ \psi^\D(\gamma_1^0 \gamma_2^0 \gamma_1^0 \gamma_2^1 - \gamma_1^0 \gamma_2^0 \gamma_1^1 \gamma_2^0) ~=~ \psi^\D (\id_2 \otimes \sigma_3 - \sigma_3 \otimes \id_2) \nonumber\\
   &\stackrel{\rm eqs.\eqref{eq:pauli},\eqref{eq:tensorrep}}{=}~ (\psi_1^*,\psi_2^*,\psi_3^*,\psi_4^*) \left(\begin{array}{cccc}
 0&~&~&~\\ 
 ~&-2&~&~\\
 ~&~&2&~\\
 ~&~&~&0
 \end{array} \right)\left(\begin{array}{c}
  \psi_1\\ 
  \psi_2\\
  \psi_3\\
  \psi_4
  \end{array} \right)\nonumber\\
 \Leftrightarrow~~~|\psi_2|^2~&=~|\psi_3|^2.
 \end{align}
 This relation is satisfied if and only if there exists a phase function $\theta$ such that $\psi_2 = e^{-i\theta} \psi_3$. Applied to eq. \eqref{eq:condj}, this yields the claim \eqref{eq:currentconsbdyconds}. \qed
 \end{enumerate}

\end{proof}
\begin{remark}
  Note that the strategy used in the proof can be generalized immediately to arbitrary particle numbers and dimensions. Furthermore, it is purely geometrical and therefore leads to Lorentz invariant conditions for the tensor current (if the domain $\Omega$ is Lorentz invariant\footnote{We call a set $A$ Lorentz invariant if for each point $p \in A$ the Lorentz transformed point $p'$ is also contained within the set.}).
\end{remark}

\section{Lorentz invariance}
In this section, we address the issue of Lorentz invariance of the constructions used in this paper. First, we state clearly our understanding of Lorentz invariance. Then we briefly review some basic representation theory of the one-dimensional proper Lorentz group and discuss the invariance of the model. The main result is the proof that the probability-conserving boundary conditions \eqref{eq:currentconsbdyconds} are indeed Lorentz invariant under certain conditions on the phase functions. We also point out a subclass of conditions for which the invariance is manifest.

\subsection{The meaning of Lorentz invariance for the model}
For the model to be Lorentz invariant, we require the following points:
\begin{enumerate}
 \item If a function $\psi$ solves the multi-time wave equations in one frame, it also solves the equations in every other frame. Furthermore, the equations have the same functional form in all frames.
 \item Probability conservation holds in all frames.
 \item In any frame, initial data can be given on $(\Sigma_t \times \Sigma_t)\cap \Omega$ where $\Sigma_t$ is an equal-time hypersurface.
 \item If a function $\psi$ satisfies the boundary conditions in one frame, it satisfies the Lorentz-transformed boundary conditions in every other frame. These boundary conditions have the same functional form in all frames.
\end{enumerate}
Before commenting on these points, we state the transformation properties in question.

\subsection{Representation of the one-dimensional proper Lorentz group} 
In $1+1$ dimensions, the proper Lorentz group $\mathcal{L}_+^\uparrow$ has only one generator, the boost generator in $z$-direction ($x_i = (t_i,z_i)$). For the spinor representation acting on the spin index of the $i$-th particle, this is:
\begin{equation}
 S_i^{01} ~=~ \tfrac{1}{4}[\gamma_i^0,\gamma_i^1].
 \label{eq:generator}
\end{equation}
A two-time wave function transforms as follows under the action of an element $\Lambda \in \mathcal{L}_+^\uparrow$:
\begin{equation}
 \psi(x_1,x_2)~~\stackrel{\Lambda}{\longmapsto}~~ \psi'(x_1,x_2) ~\equiv~ S_1[\Lambda] \, S_2[\Lambda] \, \psi(\Lambda^{-1}x_1, \Lambda^{-1}x_2)
 \label{eq:twotimespinorrep}
\end{equation}
where
\begin{equation}
 S_i[\Lambda]~=~ \exp(\beta S_i^{01}).
 \label{eq:si}
\end{equation}
Here, $\beta$ is a real parameter determined by $\Lambda$.\\
Finally, for later use, note the following relation:
\begin{equation}
 \gamma_i^\mu S_i[\Lambda] ~=~ S_i[\Lambda] \, \Lambda_\nu^\mu \gamma_i^\nu. 
 \label{eq:commutsl}
\end{equation}
The above information is already sufficient to discuss the requirements mentioned in the previous subsection:
\begin{enumerate}
 \item By the standard arguments about the Lorentz invariance of the Dirac equation (see eg. \cite{thaller}), one can show that the multi-time Dirac equations \eqref{eq:reltwotime} indeed transform covariantly. Recalling the argument in the introduction (following eq. \eqref{eq:connection}), we note that in order to discuss Lorentz invariance of the wave equations, it is crucial that a multi-time wave function is considered. Moreover, this consideration also requires the domain $\Omega$ to be Lorentz invariant. The space-like configurations $\mathscr{S}$ are of course such a Lorentz invariant set.\\
 Furthermore, under $\Lambda \in \mathcal{L}_+^\uparrow$, one obtains
 \begin{equation}
  j^{\mu \nu}(x_1,x_2) ~~\stackrel{\Lambda}{\longmapsto} ~~\Lambda_\rho^\mu \, \Lambda_\sigma^\nu \, j^{\rho \sigma}(\Lambda^{-1}x_1,\Lambda^{-1}x_2),
 \label{eq:trafoj}                                                                      
 \end{equation}
 i.e. $j^{\mu \nu}$ transforms similarly to a tensor, the only difference being the arguments in configuration space-time instead of just space-time.
 \item As shown in theorem \ref{thm:currentconsbdyconds}, if probability conservation holds on one space-like hypersurface, it holds on all space-like hypersurfaces. This, of course, includes the equal-time hypersurfaces for all frames.
 \item So far, we assumed the initial data to be given in one particular frame. However, as the choice of this frame is not fixed by any circumstance, one can simply choose the coordinates such that the initial data surface is actually of the desired form.
 \item The Lorentz invariance of the boundary conditions is the most subtle point. Because of the transformation properties of $j$, the conditions on the tensor current are easily seen to be Lorentz invariant (see eq. \eqref{eq:currentconsconds}). However, for the conditions \eqref{eq:currentconsbdyconds} on the components of $\psi$, Lorentz invariance is not manifest and the transformation behavior has to be checked explicitly.
\end{enumerate}

\subsection{Lorentz invariance of the boundary conditions}

\begin{lemma}
 The current-conserving boundary conditions
 \begin{align}
 \psi_2(t,z-0,t,z+0) ~&\stackrel{!}{=}~ e^{-i\theta_1} \, \psi_3(t,z-0,t,z+0),~~t,z \in \R,\nonumber\\
 \psi_2(t,z+0,t,z-0) ~&\stackrel{!}{=}~ e^{-i\theta_2} \, \psi_3(t,z+0,t,z-0),~~t,z \in \R
 \label{eq:libdyconds2}
 \end{align}
 are Lorentz invariant if the functions $\theta_1, \theta_2$ transform as Lorentz scalars, i.e. if
 \begin{equation}
  \theta_i(t,z) ~\stackrel{\Lambda}{\longmapsto}~\theta_i(\Lambda^{-1}(t,z))~\forall \Lambda \in \mathcal{L}_+^\uparrow,~i = 1,2.
  \label{eq:lscalar}
 \end{equation}
 \label{thm:libdyconds}
\end{lemma}
\begin{proof}
 We explicitly determine the transformation properties of the components $\psi_i$. According to eq. \eqref{eq:twotimespinorrep}, we need to calculate the matrices $S_1[\Lambda], S_2[\Lambda]$ via formula \eqref{eq:si}. We have:
 \begin{align}
  S_1[\Lambda]~&=~ \exp(\beta \tfrac{1}{2}\gamma_1^0 \gamma_1^1)~\stackrel{\rm eq. \eqref{eq:gamma1d}}{=}~ \exp(\beta \tfrac{1}{2} \sigma_3 \otimes \id_2)\nonumber\\
              &\stackrel{\rm eq. \eqref{eq:tensorrep}}{=}~ \sum_{k=0}^\infty \frac{(\beta/2)^k}{k !} \left(\begin{array}{cc}
\id_2&~\\
~&(-\id_2)^k
\end{array} \right)\nonumber\\
 &=~ \cosh \beta \cdot \,  \id_4 + \sinh \beta  \left(\begin{array}{cc}
\id_2&~\\
~&-\id_2
\end{array} \right),
  \label{eq:explcitsi1}
 \end{align}
\begin{align}
  S_2[\Lambda]~&=~ \exp(\beta \tfrac{1}{2}\gamma_2^0 \gamma_2^0)~\stackrel{\rm eq. \eqref{eq:gamma1d}}{=}~ \exp(\beta \tfrac{1}{2} \id_2 \otimes \sigma_3)\nonumber\\
              &\stackrel{\rm eq. \eqref{eq:tensorrep}}{=}~ \sum_{k=0}^\infty \frac{(\beta/2)^k}{k !} \left(\begin{array}{cccc}
1&~&~&~\\ 
~&(-1)^k&~&~\\
~&~&(-1)^k&~\\
~&~&~&1
\end{array} \right)\nonumber\\
 &=~ \cosh \beta \cdot \id_4 + \sinh \beta  \left(\begin{array}{cccc}
1&~&~&~\\ 
~&-1&~&~\\
~&~&-1&~\\
~&~&~&1
\end{array} \right).
  \label{eq:explcitsi2}
 \end{align}
It follows that:
\begin{align}
  S_1[\Lambda] S_2[\Lambda]~&=~ \cosh^2 \beta \id_4 + 2 \cosh \beta \sinh \beta \left(\begin{array}{cccc}
1&~&~&~\\ 
~&0&~&~\\
~&~&0&~\\
~&~&~&-1
\end{array} \right) + \sinh^2 \beta \left(\begin{array}{cccc}
1&~&~&~\\ 
~&-1&~&~\\
~&~&-1&~\\
~&~&~&1
\end{array} \right).
  \label{eq:explcitsi3}
 \end{align}
Thus
\begin{align}
  \psi_i(x_1,x_2)~\stackrel{\Lambda}{\longmapsto}~ \psi_i(\Lambda^{-1} x_1, \Lambda^{-1}x_2)~~{\rm for}~i = 2,3.
  \label{eq:explcitsi4}
 \end{align}
Using this transformation property in eq. \eqref{eq:libdyconds2} together with eq. \eqref{eq:lscalar} immediately yields the claim.\qed
\end{proof}

\begin{lemma}
 In case of $e^{-i\theta_k(t,z)} \equiv \pm i$, the boundary conditions \eqref{eq:libdyconds2} can be rewritten in the following manifestly Lorentz invariant form:
  \begin{align}
 \varepsilon_{\mu \nu} \gamma_1^\mu \gamma_2^\nu \psi(t,z-0,t,z+0) ~&\stackrel{!}{=}~ \pm i ( \id_4 + \gamma_1^5 \gamma_2^5)\psi(t,z-0,t,z+0),~~t,z \in \R,\nonumber\\
 \varepsilon_{\mu \nu} \gamma_1^\mu \gamma_2^\nu \psi(t,z+0,t,z-0) ~&\stackrel{!}{=}~ \pm i ( \id_4 + \gamma_1^5 \gamma_2^5)\psi(t,z+0,t,z-0),~~t,z \in \R
  \label{eq:libdyconds3}
 \end{align}
  where
  \begin{equation}
   \gamma_k^5 ~:=~ i \gamma_k^0 \gamma_k^1,~~k = 1,2.
   \label{eq:gamma5}
  \end{equation}
 \label{thm:manifestlibdyconds}
\end{lemma}

\begin{proof}
 \begin{align}
  \varepsilon_{\mu \nu} \gamma_1^\mu \gamma_2^\nu ~&=~  \gamma_1^0 \gamma_2^1 - \gamma_1^1 \gamma_2^0 ~=~  \sigma_1 \otimes \id_2 \cdot \id_2 \otimes (\sigma_1 \sigma_3) - (\sigma_1 \sigma_3) \otimes \id_2 \cdot \id_2 \otimes \sigma_1~=~\left(\begin{array}{cccc}
~&~&0&0\\ 
~&~&2&0\\
0&-2&~&~\\
0&0&~&~
\end{array} \right), \nonumber\\
 \id_4 + \gamma_1^5 \gamma_2^5 ~&=~ \id_4 + i\sigma_3 \otimes \id_2 \cdot i \id_2 \otimes \sigma_3~=~  \id_4 - \left(\begin{array}{cccc}
1&~&~&~\\ 
~&-1&~&~\\
~&~&-1&~\\
~&~&~&1
\end{array} \right) ~=~ \left(\begin{array}{cccc}
0&~&~&~\\ 
~&2&~&~\\
~&~&2&~\\
~&~&~&0
\end{array} \right).
  \label{eq:explicitreps}
 \end{align}
Thus, we obtain:
\begin{align}
 \varepsilon_{\mu \nu} \gamma_1^\mu \gamma_2^\nu \psi ~&\stackrel{!}{=}~ \pm i ( \id_4 + \gamma_1^5 \gamma_2^5) \psi\nonumber\\
\Leftrightarrow~~~ \left(\begin{array}{cccc}
~&~&0&0\\ 
~&~&2&0\\
0&-2&~&~\\
0&0&~&~
\end{array} \right) \left(\begin{array}{c}
  \psi_1\\ 
  \psi_2\\
  \psi_3\\
  \psi_4
  \end{array} \right) ~&\stackrel{!}{=}~ \pm i \, \left(\begin{array}{cccc}
0&~&~&~\\ 
~&2&~&~\\
~&~&2&~\\
~&~&~&0
\end{array} \right) \left(\begin{array}{c}
  \psi_1\\ 
  \psi_2\\
  \psi_3\\
  \psi_4
  \end{array} \right).
 \label{eq:equivconds}
\end{align}
This is in turn equivalent to the following conditions
\begin{align}
 0 ~&=~ 0, \nonumber\\
 \psi_3 ~&=~ \pm i \psi_2,\nonumber\\
 - \psi_2 ~&=~ \pm i \psi_3,\nonumber\\
 0 ~&=~ 0.
 \label{eq:equivconds2}
\end{align}
This yields the claim. \qed
\end{proof}

\begin{remark}
 For the manifestly Lorentz invariant boundary conditions \eqref{eq:libdyconds3}, one can use the usual representation-independent strategy to prove Lorentz invariance:\\[0.2cm]
 Assume that the conditions are fulfilled in one frame $F$. Now consider the same conditions in another frame $F'$ that is connected with the former one by a Lorentz transformation $\Lambda$. We have to show that \eqref{eq:libdyconds3} is satisfied as a consequence of the transformation law \eqref{eq:twotimespinorrep} for $\psi$ as well as \eqref{eq:libdyconds3} for $F$. Consider eq. \eqref{eq:libdyconds3} for $F'$:
\begin{align}
  \varepsilon_{\mu \nu} \gamma_1^\mu \gamma_2^\nu \psi'(x_1,x_2) ~&\stackrel{!}{=}~ \pm i ( \id_4 + \gamma_1^5 \gamma_2^5) \psi'(x_1,x_2) \nonumber\\
 \Leftrightarrow~~~\varepsilon_{\mu \nu} \gamma_1^\mu \gamma_2^\nu S_1[\Lambda] S_2[\Lambda] \, \psi(\Lambda^{-1}x_1,\Lambda^{-1}x_2) ~&\stackrel{!}{=}~ \pm i ( \id_4 + \gamma_1^5 \gamma_2^5) S_1[\Lambda] S_2[\Lambda] \, \psi(\Lambda^{-1}x_1,\Lambda^{-1}x_2)
\end{align}
where $x_1 = (t,z\pm 0)$ and $x_2 = (t,x_2\mp 0)$. As $S_1[\Lambda]$ and $S_2[\Lambda]$ are invertible and because $(\Lambda^{-1}x_1,\Lambda^{-1}x_2)$ again has the form $(t',z'\pm 0,t',z'\mp 0)$, it is sufficient to prove that $S_1[\Lambda] S_2[\Lambda]$ commutes with both $\varepsilon_{\mu \nu} \gamma_1^\mu \gamma_2^\nu$ as well as $\id_4 + \gamma_1^5 \gamma_2^5$. Consider first:
\begin{align}
 \varepsilon_{\mu \nu} \,\gamma_1^\mu \gamma_2^\nu \, S_1[\Lambda] S_2[\Lambda] ~\stackrel{\rm eq.\eqref{eq:commutsl}}{=}~ S_1[\Lambda]  S_2[\Lambda]\, \varepsilon_{\mu \nu} \, \Lambda_\rho^\mu \Lambda_\sigma^\nu \,  \gamma_1^\rho  \gamma_2^\sigma =~ S_1[\Lambda]  S_2[\Lambda]\, \varepsilon_{\rho \sigma} \,  \gamma_1^\rho  \gamma_2^\sigma
\end{align}
where in the equality we used $\det(\Lambda) \varepsilon_{\rho \sigma} = \varepsilon_{\mu \nu} \, \Lambda_\rho^\mu \Lambda_\sigma^\nu$ as well as $\det(\Lambda) = 1$ for $\Lambda \in \mathcal{L}_+^\uparrow$.\\
In order to show that $S_1[\Lambda] S_2[\Lambda]$ commutes with $\id_4 + \gamma_1^5 \gamma_2^5$ note that it is sufficient that the generators $S_k^{01},~k = 1,2$ commute with $\gamma_1^5 \gamma_2^5$. We have: $S_k^{01} \gamma_j^5 = \gamma_j^5 S_k^{01},~j,k = 1,2$. For $j \neq k$, this is obvious and in case $j = k$ the equation easily follows from $S_k^{01} = \tfrac{1}{2} \gamma_k^0 \gamma_k^1$ (see eq. \eqref{eq:generator}) as well as $\gamma_k^5 = i \gamma_k^0 \gamma_k^1$ (see eq. \eqref{eq:gamma5}).
\end{remark}

\section{Interaction}
In this section, we analyze the physical meaning of the boundary conditions \eqref{eq:libdyconds2} and in particular the question if they lead to interaction. In order to address this question appropriately, we suggest a simple and clear-cut notion of interaction. Then we use the explicit solution of our model to determine the time evolution of certain wave packets for which interaction effects are clearly visible. With this result we gain physical insight into the detailed nature of the time evolution implied by our model. Moreover, we can use the result to conclude that it indeed leads to interaction.

\paragraph{A criterion for interaction:} Most often, ``interaction'' in quantum mechanics is simply defined by the presence of an interaction potential in the Hamiltonian. This notion of interaction is obviously not adequate for models such as ours where one aims at implementing interaction effects via boundary conditions. A more general criterion is needed:\\
A quantum-mechanical model is called \textit{free} if every initial product wave function (in the particle coordinates and spin indices) remains a product wave function during time evolution. It is called \textit{interacting} if there exist initial product wave functions that do not stay product wave functions during time evolution.

\paragraph{Evolution of an initially well-localized product wave function:} 
A product wave function $\psi = \phi \otimes \chi$, where $\phi, \chi : \R^2 \rightarrow \C^2$ are two-component spinors, has the following components:
\begin{equation}
 \psi_1 = \phi_1 \chi_1,~~~\psi_2 = \phi_1 \chi_2,~~~\psi_3 = \phi_2 \chi_1,~~~\psi_3 = \phi_2 \chi_2.
 \label{eq:productfncomponents}
\end{equation}
In order to make the example as simple as possible, we choose the initial wave function $\psi(0,z_1,0,z_2) = \phi(z_1) \otimes \chi(z_2)$ as follows:
\begin{equation}
 \phi_2 \equiv \chi_1 \equiv 0,~~~\phi_1 = \tilde{\phi}\, 1_{[a,b]},~~~\chi_2 = \tilde{\chi}\, 1_{[c,d]}
 \label{eq:initialprod}
\end{equation}
where $\tilde{\phi}, \tilde{\chi} : \R \rightarrow \C$ are smooth functions with support in $[a,b]$ and $[c,d]$, respectively. We choose $a < b < c < d$. $1_{[x,y]}$ denotes the characteristic function of the interval $[x,y]$, i.e. $1_{[x,y]}(z) = 1$ if $z \in [x,y]$ and $0$ otherwise. We have multiplied $\tilde{\phi}, \tilde{\chi}$ with the characteristic functions of their support to make more explicit when they vanish.\\
As a consequence of eq. \eqref{eq:initialprod}, we have:
\begin{align}
 &\psi_1(0,z_1,0,z_2) ~=~ \psi_3(0,z_1,0,z_2) ~=~ \psi_4(0,z_1,0,z_2) ~=~ 0,  \nonumber\\
 &\psi_2(0,z_1,0,z_2) ~=~ \tilde{\phi}(z_1) \, \tilde{\chi}(z_2)\,  1_{[a,b]}(z_1) \, 1_{[c,d]}(z_2). \label{eq:initialprod3}
\end{align}
Specifically, we note that $\psi_2(0,z_1,0,z_2) = 0~{\rm for}~ z_2 \geq z_1$. Therefore, $\psi$ satisfies the boundary conditions \eqref{eq:libdyconds2} in the form $0 = 0$.\\
\begin{lemma}
The solution of the IBVP defined by eqs. \eqref{eq:initialprod3}, \eqref{eq:currentconsbdyconds} is given by:
\begin{align}
 &\psi ~\equiv~ 0~~~{\rm on}~ \Omega_2,\nonumber\\
 &\psi_1 ~\equiv~ \psi_4 ~\equiv~ 0~~~{\rm on}~ \Omega, \nonumber\\
 &\psi_2(t_1,z_1,t_2,z_2)~=~\tilde{\phi}(z_1-t_1) \tilde{\chi}(z_2+t_2) \, 1_{[a,b] + t_1}(z_1) \,  1_{[c,d]-t_2}(z_2)\nonumber\\
&~~~~~~~~~~~~~~~~~~~~~~~~~\times \Theta(-z_1+t_1+z_2+t_2)~~~{\rm on}~ \Omega_1, \nonumber\\
 &\psi_3(t_1,z_1,t_2,z_2)~=~e^{i\theta_1((z_1-z_2+t_1+t_2)/2, (z_1+z_2+t_1-t_2)/2)}\tilde{\phi}(z_2-t_2) \tilde{\chi}(z_1+t_1)\nonumber\\
 &~~~~~~~~~~~~~~~~~~~~~~~~~\times 1_{[a,b] + t_2}(z_2) \,  1_{[c,d]-t_1}(z_1) \, \Theta(z_1+t_1-z_2+t_2)~~~{\rm on}~ \Omega_1
 \label{eq:timevolexample}
\end{align}
 where $\Theta$ is the Heaviside function.
\end{lemma}
\begin{remark}
 Note that one can leave away the $\Theta$-functions from the equations for $\psi_2$, $\psi_3$ as they only set the functions to zero where they vanishes anyway.
\end{remark}
\begin{proof}
 We make use of the explicit solution \eqref{eq:explicitsolution}, the only difference being that the functions $h_1^{\pm}$ are given by $\psi_2,\psi_3$ via the boundary conditions \eqref{eq:currentconsbdyconds}. Evidently, $\psi \equiv 0$ on $\Omega_2$, so we can focus on $\Omega_1$.
 \begin{enumerate}
  \item For $z_1 + t_1 \geq z_2-t_2$, $\psi_2$ is determined by initial data and we have:
  \begin{align}
   \psi_3(t_1,z_2,t_2)~&=~ h_1^+(\underbrace{\tfrac{z_1-z_2+t_1+t_2}{2}}_{t},\underbrace{\tfrac{z_1+z_2+t_1-t_2}{2}}_{z} )\nonumber\\
   &=~ e^{i\theta_1(t,z)} \psi_2(t,z-0,t,z+0) ~=~ e^{i\theta_1(t,z)} g_2^{(1)}(z-t,t+z)\nonumber\\
   &=~ e^{i\theta_1((z_1-z_2+t_1+t_2)/2,(z_1+z_2+t_1-t_2)/2)} g_2^{(1)}(z_2-t_2,z_1+t_1).
   \label{eq:timevolderiv1}
  \end{align}
  \item Similarly, for $z_1-t_1\geq z_2+t_2$, $\psi_3$ is determined by initial data and we obtain:
  \begin{align}
   \psi_2(t_1,z_2,t_2)~&=~ h_1^-(\underbrace{\tfrac{-z_1+z_2+t_1+t_2}{2}}_{t'},\underbrace{\tfrac{z_1+z_2-t_1+t_2}{2}}_{z'} )\nonumber\\
   &=~ e^{i\theta_1(t',z')} \psi_2(t',z'-0,t',z'+0) ~=~ e^{i\theta_1(t',z')} g_2^{(1)}(z'-t',t'+z')\nonumber\\
   &=~ e^{-i\theta_1((-z_1+z_2+t_1+t_2)/2,(z_1+z_2-t_1+t_2)/2)} g_3^{(1)}(z_2+t_2,z_1-t_1).
   \label{eq:timevolderiv2}
  \end{align}
 \end{enumerate}
Noting that the two cases are exclusive on $\Omega_1$, we have determined $\psi$ uniquely (see eq. \eqref{eq:explicitsolution}). Reading off the initial data from eq. \eqref{eq:initialprod3} and using the explicit solution \eqref{eq:explicitsolution}, the claim follows. \qed
\end{proof}
Let us come back to the interaction criterion. If $\psi$ were a product function for all times, we would have to be able to factorize it analogously to eq. \eqref{eq:initialprod}. However, we can see from eq. \eqref{eq:timevolexample} that this is not possible: the phase function in the expression for $\psi_3$ in general contains the variables $t_1,z_1,t_2,z_2$ in an inextricable way. Even if $\theta_1$ were to decompose into a sum of functions of $t_1,z_1$ and $t_2,z_2$, respectively, a slightly modified example with non-vanishing initial $\psi_3$-component would show an interaction effect. Then the contribution to $\psi_3$ in eq. \eqref{eq:timevolexample} would appear additively and thus produce entanglement.\\
 We have therefore found an example for an initial product wave function which becomes entangled with time and obtain the following result:
\begin{theorem}
Our model, defined by eqs. \eqref{eq:ibvp1}, \eqref{eq:ibvp2} and \eqref{eq:libdyconds2}, is interacting in the sense of the criterion presented above.
 \label{thm:interaction}
\end{theorem}
\begin{remark}
 \begin{enumerate}
  \item One can see from the derivation of eq. \eqref{eq:timevolexample} that a similar interaction effect is present for all initial wave functions with $\psi_2 \neq 0$.
  \item Example \eqref{eq:timevolexample} allows us to understand the interaction in a more detailed way (see fig. \ref{fig:interaction}). Focusing on times $t_1 = t_2 = t$, we see that at $t = 0$ each particle has an associated wave packet localized in a certain region. These regions do not overlap. For $t > 0$, the wave packets are translated towards each other with speed $c = 1$, so that the gap between them shrinks with speed 2. There is no dispersion in the mass-less case. Of course, these wave packets are not actually functions on physical space but on different copies thereof, as factors of configuration space. As soon as they would meet, a scattering process happens: the wave packets swap place, i.e. the one associated with particle 1 becomes associated with particle 2 and the other way around. This process is associated with a phase. They move in opposite directions as before -- and the corresponding contribution to the wave function is one associated with different spin. In summary, the interaction 
has range zero, respects retardation and produces a scattering from one spin component into another.
 \item Due to this behavior, the model only describes scattering processes. ``Bound states'' or ``resonances'' do not appear. To define these concepts rigorously, one can use the single-time model obtained by restricting the multi-time wave function to a common time.
 \end{enumerate}
\end{remark}

\begin{figure}[p]
   \centering
   \includegraphics[width=\textwidth]{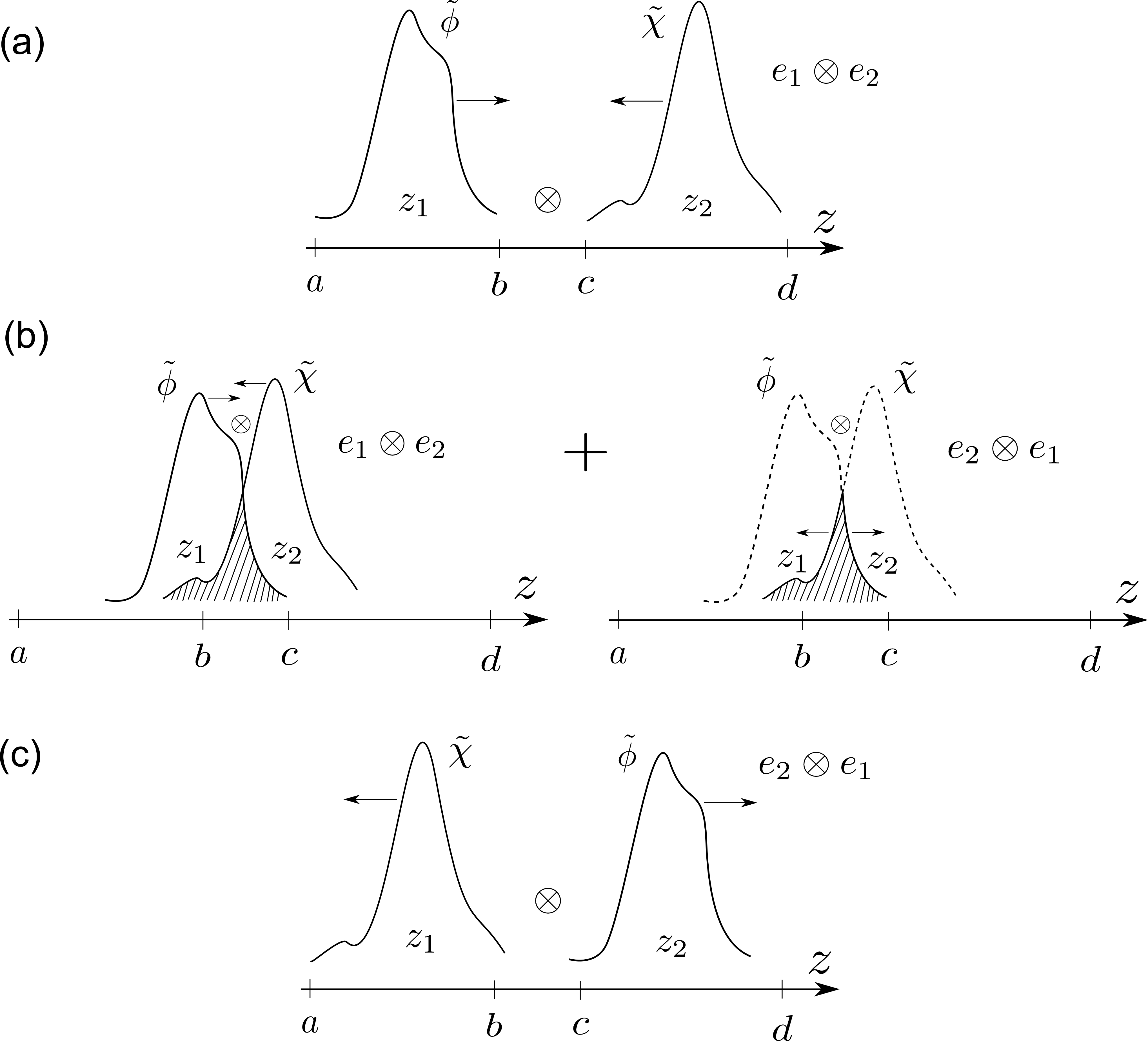}
   \caption{Schematic illustration of the interaction.
   (a) $t_1 = t_2 = 0$: The wave packets move towards each other (in different parts of configuration space) with speed $c = 1$ and without dispersion. $\tilde{\varphi}$ is associated with particle 1 and $\tilde{\chi}$ with particle 2. The only non-zero component of the total wave function is $\psi_2$ (associated with $e_1 \otimes e_2$).
   (b) $(d-a)/2 > t_1 = t_2 > (c-b)/2$: The wave packets overlap (if plotted in the same space). The hatched area for $e_1 \otimes e_2$ has left $\Omega_1$. It reappears with a phase in the component for $e_2 \otimes e_1$. The wave packets have swapped place. The hatched parts of $\tilde{\varphi}$ are now associated with particle 2 and the hatched parts of $\tilde{\chi}$ with particle 1.
   (c) $t_1 = t_2 > (d-a)/2$: The wave packets have passed each other. The only non-zero component is $\psi_3$ (associated with $e_2 \otimes e_1$). }
 \label{fig:interaction}
 \end{figure}

\section{Indistinguishable particles}
If we describe indistinguishable fermions, the correct transformation properties of our wave function are\footnote{Eq. \eqref{eq:antisymmproperty} is a straightforward generalization of the well-known antisymmetry properties of a single-time wave function, i.e. $\varphi^{s_1 s_2}(\mathbf{x}_1, \mathbf{x}_2,t) =- \varphi^{s_2 s_1}(\mathbf{x}_2, \mathbf{x}_1,t)$.}
\begin{equation}
 \psi^{s_1 s_2}(x_1,x_2)~=~ -\psi^{s_2s_1}(x_2,x_1).
 \label{eq:antisymmproperty}
\end{equation}
Here, the double index $s_1 s_2$ where $s_1,s_2 = -1,1$ is a different way of denoting the spin components. We have:
\begin{equation}
 \psi^{-1 -1} ~=~ \psi_1,~~~\psi^{-1 1} ~=~ \psi_2,~~~\psi^{1 -1} ~=~ \psi_3,~~~\psi^{1 1} ~=~ \psi_4.
 \label{eq:spincomponents}
\end{equation}
One may ask: is our model compatible with these transformation rules? In order to answer this question, we note that given a solution of the IBVP with boundary conditions \eqref{eq:currentconsbdyconds} on $\Omega$, we may use eq. \eqref{eq:antisymmproperty} to continue it antisymmetrically to $\Omega_2$. It is easy to see that as a consequence of the two-time Dirac equation \eqref{eq:twotime} on $\Omega_1$ it also solves the two-time Dirac equation on $\Omega_2$, with initial data that are the antisymmetric continuation of those on $\Omega_1$. However, under which circumstances are the boundary conditions on $\Omega_2$ satisfied?\\
To answer this question, we consider the transformation behavior of the boundary conditions \eqref{eq:currentconsbdyconds} on $\Omega_1$:
\begin{align}
 \psi_2(t,z-0,t,z+0) ~&\stackrel{!}{=}~ e^{-i \theta_1(t,z)} \psi_3(t,z-0,t,z+0),~~t,z \in \R\nonumber\\
 \stackrel{\rm antisym.}{\longmapsto}~~~-\psi_3(t,z+0,t,z-0) ~&\stackrel{!}{=}~ -e^{-i \theta_1(t,z)} \psi_2(t,z+0,t,z-0),~~t,z \in \R\nonumber\\
 \Leftrightarrow~~~\psi_2(t,z+0,t,z-0) ~&\stackrel{!}{=}~ e^{+i \theta_1(t,z)} \psi_2(t,z+0,t,z-0),~~t,z \in \R\nonumber
  \label{eq:bdyantisymtrafo}
\end{align}
which is a boundary condition on $\Omega_2$. Comparison with eq. \eqref{eq:currentconsbdyconds} yields the following result:
\begin{lemma}
 The IBVP defined by eqs. \eqref{eq:ibvp1}, \eqref{eq:ibvp2} and \eqref{eq:currentconsbdyconds} is compatible with antisymmetry under particle exchange \eqref{eq:antisymmproperty} if the initial data are antisymmetric and if $\theta_2 = -\theta_1$.
 \label{thm:antisymmetry}
\end{lemma}

\section{Discussion and outlook}
Starting from natural physical considerations, we were led to the view that relativistic quantum mechanics should build around the notion of a multi-time wave function. To consider a multi-time wave function seems necessary to discuss the Lorentz invariance of the theory. This viewpoint raised several mathematical questions: The evolution equations are in general overdetermined and the corresponding consistency condition excludes interaction potentials. This motivated the search for different mechanisms of interaction. In this paper, we studied the possibility of interactions by boundary conditions on configuration space-time. The natural domain in such a view are the space-like configurations. We chose the simplest model possible that still possesses essential properties like Lorentz invariant wave equations and a conserved current with positive density: mass-less Dirac equations for two particles in one spatial dimension. The study of existence and uniqueness of solutions for multi-time wave equations on 
domains with boundary, however, required to go beyond the usual functional-analytic setting of a unitary group (or a multi-time version thereof). The relative simplicity of our model allowed for an alternative strategy to prove the existence and uniqueness of solutions: a generalized version of the method of characteristics. After finding an appropriate notion of relativistic probability conservation, we employed a geometrical construction involving Stokes' theorem to extract a class of boundary conditions guaranteeing this property. In addition, we proved that the theory with said class of boundary conditions is Lorentz invariant in a strict sense. Furthermore, we showed that the model is interacting in an appropriate sense. The details of the interaction were studied at the example of an initially well-localized wave packet for each of the two particles. It was found that the model describes a scattering process with range zero and associated with a spin flip. Finally, we analyzed the requirements of 
antisymmetry for indistinguishable particles and showed that they lead to a further selection of the class of physically sensible boundary conditions.\\
We emphasize that the model has been constructed in a way which is consistent with realistic relativistic quantum theories such as the hypersurface Bohm-Dirac models \cite{hbd} and relativistic GRW theories \cite{grwf,rel_grw}. The main requirement of these theories is a multi-time formulation with a conserved tensor current like in eq. \eqref{eq:j}.\\
The reader may also be interested in a comparison with the more familiar single-time picture. This is always possible by restricting the multi-time wave function to a single global time via eq. \eqref{eq:connection}. Then the multi-time model \eqref{eq:model} yields a single-time model (using the chain rule to obtain a single wave equation from the multi-time wave equations). This model can be analyzed on its own terms via a functional-analytic approach, using the methods of zero-range physics. This changes the notion of a solution from ``classical'' to ``weak''. Nevertheless, we expect such an approach to lead to results for the subclass of classical solutions similar to the results for the solutions in our model when evaluated at equal times. For the functional-analytic approach, one expects that the phase functions in eq. \eqref{eq:currentconsbdyconds} should not depend on $t$ -- and therefore, by Lorentz invariance, not on $z$, either. Unsurprisingly, our approach is slightly more general in this regard, 
as it is specifically designed for an equal treatment of space and time variables (see the introduction).\\
In view of the success of the methods for this very simple model, it is natural to ask for possible generalizations with respect to several aspects:
\begin{itemize}
 \item \textbf{Non-zero masses:} The study of the mass-less case is mainly a technical simplification. The conservation properties of the tensor current as well as the derivation of the class of probability-conserving boundary conditions are independent of the presence of mass terms in the multi-time equations. However, the strategy of the existence and uniqueness proof was to make use of the fact that the solution has to be constant along the multi-time characteristics. This is not true anymore for the case with mass. Preliminary investigations have led us to the idea that it might be possible to reformulate the simultaneous differential systems of multi-time equations into a single system of multi-time integral equations. For these integral equations, fixed point arguments could be used to prove the existence and uniqueness of a solution.
 \item \textbf{$N>2$ particles:} The generalization for $N$ particles should be straightforward and is currently under investigation. Many of the constructions in this paper such as the domain, the multi-time equations and the current form are immediately extendible to the multi-particle case. The idea is to prescribe boundary conditions on the set where a pair of particles is at the same space-time point and the other particles space-like to this point. Compared with the two-particle case, this raises additional questions if these boundary conditions do not overdetermine the problem. Also, in the study of existence and uniqueness of solutions, it is now possible that components of the wave function are determined via successive collisions, meaning that one component is determined by initial data, determines another one via boundary conditions and this other one determines yet another via other boundary conditions and so on.
 \item \textbf{Higher dimensions:} An immediate generalization of the model to $d>1$ is not feasible. One can see this from the following consideration: the boundary conditions as the mechanism of interaction are derived from the fact that the integral $\int_{\mathscr{C}} \omega_j$ has to vanish to ensure probability conservation. However, $\omega_j$ is in general a $2d$-form and $\mathscr{C} = \{(x_1,x_2) \in \R^{1+d}\times \R^{1+d} : x_1 = x_2\}$ is $(1+d)$-dimensional. Thus, $\mathscr{C}$ is a zero-measure set for $d > 1$ and the integral vanishes without boundary conditions. In fact, one can even show that $\int_{(\Sigma \times \Sigma)\cap \Omega} \omega_j$ is a so-called ``energy integral''\footnote{This means that if $\int_{(\Sigma \times \Sigma)\cap \Omega} \omega_j$ with $\omega_j$ derived from $\psi_1 - \psi_2$ according to eq. \eqref{eq:defcurrentform} is zero, one can conclude that $\psi_1 \equiv \psi_2$ on $(\Sigma \times \Sigma)\cap \Omega$ in a suitable (weak) sense.} for the multi-time 
equations and that therefore probability conservation guarantees uniqueness of solutions. Since this integral is automatically conserved for $d > 1$, no boundary conditions are required from a mathematical perspective. Prescribing boundary conditions in spite of this would either influence the wave function only on a low-dimensional set or lead to (possibly complicated) restrictions on the initial data. Without a clear physical reason for conditions of this kind, this option does not seem sensible. We therefore conclude that our construction has to be modified in order to produce interaction effects for $d > 1$.
 \item \textbf{Different domains:} Motivated by the fact that $\int_{\mathscr{C}} \omega_j$, the flux through the boundary, vanishes for $d > 1$, one can try to find a different Lorentz invariant domain that yields a non-vanishing flux through the boundary, i.e. $\int_{\mathscr{B}} \omega_j$ where $\mathscr{B}$ has a dimension of at least $2d$. Such a domain is for example given by the space-like configurations with a minimum space-like distance $\alpha$:
 \[ \Omega_\alpha = \{ (x_1,x_2) \in \R^{1+d}\times \R^{1+d} : (x_1^0-x_2^0)^2 - (\mathbf{x}_1- \mathbf{x}_2)^2 < -\alpha^2\}.\]
 For the one-dimensional case one could proceed similarly as before, just with a more complicated geometry. The main question is if the boundary conditions obtained from demanding $\int_{\mathscr{B}} \omega_j \stackrel{!}{=} 0$ do not overdetermine the problem. For $d>1$ one would have to devise a new strategy to analyze the question of the existence of a solution, as the previous method is based on the possibility to simultaneously diagonalize all the matrices in the multi-time equations (at least those in the highest order terms). This is, of course, not possible for the Dirac equation for $d >1$. However, on the physical side this idea is not fully satisfactory because one introduces an additional constant $\alpha$ without an apparent deeper reason.
\end{itemize}

\subsection*{Acknowledgments}
I am grateful for helpful discussions with Alain Bachelot, Felix Finster, Alessandro Michelangeli, Lukas Nickel, Felix Otto, Sören Petrat, Peter Pickl, Stefan Teufel and Roderich Tumulka. Special thanks go to Detlef Dürr for many insightful discussions and remarks. Financial support by the German National Academic Foundation is gratefully acknowleged.


\end{document}